\documentclass[a4paper,UKenglish,cleveref, autoref, thm-restate]{lipics-v2021}

\hideLIPIcs  

\usepackage{booktabs} 

\newcolumntype{Y}{>{\raggedright\arraybackslash}X}

\DeclareMathOperator{\rank}{rank} \DeclareMathOperator{\select}{select}
\DeclareMathOperator{\access}{access}

\newcommand{\BTHeight}[1]{H_{#1}}
\newcommand{\BlockTime}[1]{\tau_{#1}}
\newcommand{\Aset}[1]{\mathcal{A}_{#1}}

 \newcounter{dknotecounter}

\title{Online Computation of the Longest Repeating Suffix and Smallest Suffixient Sets via Incremental Run-Length BWT-based Indexes} 

\titlerunning{Online Computation of the Longest Repeating Suffix and Smallest Suffixient Sets} 

\author{Paola Bonizzoni}{Department of Computer Science, University of Milano-Bicocca, Italy \and \url{https://algolab.eu/people/bonizzoni/}}{paola.bonizzoni@unimib.it}{https://orcid.org/0000-0001-7289-4988}{}

\author{Younan Gao}{Department of Computer Science, University of Milano-Bicocca, Italy \and \url{https://web.cs.dal.ca/~younan/}}{younan.gao@unimib.it}{https://orcid.org/0000-0003-4984-2551}{}

\author{Dominik K\"{o}ppl}{ University of Yamanashi, Kofu, Japan \and \url{https://dkppl.de/}}{dkppl@yamanashi.ac.jp}{https://orcid.org/0000-0002-8721-4444}{JSPS KAKENHI Grant Number 25K21150}

\author{Gregory Kucherov}{LIGM, CNRS and Gustave Eiffel University, Marne-la-Vallée, France \and \url{http://igm.univ-mlv.fr/~koutcher/}}{gregory.kucherov@univ-eiffel.fr}{https://orcid.org/0000-0001-5899-5424}{}

\authorrunning{P. Bonizzoni, Y. Gao, D. K\"{o}ppl, and G. Kucherov}

\Copyright{Paola Bonizzoni, Younan Gao, Dominik K\"{o}ppl, and Gregory Kucherov}



\category{} 

\relatedversion{} 

\funding{Paola Bonizzoni and Younan Gao have received funding from the European Union’s Horizon 2020 research and innovation programme under the Marie Sk{\l}odowska-Curie grant agreement PANGAIA No.~872539, as well as from grant MIUR 2022YRB97K (PINC, \emph{Pangenome Informatics: From Theory to Applications}), funded by the European Union under the NextGenerationEU programme, Mission~4.}

\nolinenumbers 

\EventEditors{John Q. Open and Joan R. Access}
\EventNoEds{2}
\EventLongTitle{42nd Conference on Very Important Topics (CVIT 2016)}
\EventShortTitle{CVIT 2016}
\EventAcronym{CVIT}
\EventYear{2016}
\EventDate{December 24--27, 2016}
\EventLocation{Little Whinging, United Kingdom}
\EventLogo{}
\SeriesVolume{42}
\ArticleNo{23}

\begin{document}

\maketitle

\begin{abstract}
We revisit the online construction of \emph{smallest suffixient sets} and the online computation of the \emph{longest repeating suffix} (LRS). We give the first compressed-space online construction of smallest suffixient sets, and present two space-time trade-offs for both problems:
\begin{itemize}
\item $O(r\log n+n)$ bits of working space and $O(\log^2 n/\log \log n)$ worst-case time per character, and
\item $O(r\log n+n \log \log n)$ bits of working space and $O((\log n/\log \log n)^2)$ worst-case time per character.
\end{itemize}
Here, $r$ is the number of runs in the Burrows-Wheeler transform of the reverse of $T[1..n]$. In particular, for highly repetitive texts satisfying $r=O(n/\log n)$, the first trade-off uses $O(n)$ bits of working space, while the second uses $O(n\log\log n)$ bits.

We also prove that any deterministic online algorithm for computing LRS requires \(\Omega(n)\) bits of peak working space in the worst case, even over a constant-size alphabet. Through reductions from online LRS computation, we extend this lower bound to deterministic online algorithms maintaining either an arbitrary smallest suffixient set augmented with the length of the supermaximal right extension represented by each selected position, or the position-only smallest suffixient set obtained by selecting the rightmost occurrence of every such extension.

For constructing smallest suffixient sets, our algorithms are the first online solutions using compressed working space, improving the $O(n)$-word space required by previous online constructions.
For compressed-space online LRS computation, compared with the algorithm of Prezza and Rosone~[CiE 2020], our bounds improve their $O(\log^2 n)$ amortized time per character by factors of $\Theta(\log\log n)$ and $\Theta((\log\log n)^2)$, respectively, while also providing worst-case guarantees. 
\end{abstract}

\section{Introduction}

Computing the \emph{longest repeating suffix} (LRS) is a classical problem in which for a text $T[1..n]$, the value $LRS[i]$ stores the length of the longest suffix of $T[1..i]$ that occurs at least twice in $T[1..i]$. This information is useful when solving pattern analysis and data compression~\cite{ZivL77}, especially by analyzing repetitions in the text~\cite{KolpakovK99}, and has applications in deriving the Lempel-Ziv factorization~\cite{LempelZ76}.

\emph{Smallest suffixient sets}~\cite{cenzato2024suffixient} were recently introduced as a measure of string repetitiveness, motivated by the study of highly repetitive data such as genomic collections. The definition of suffixient sets is based on the notion of \emph{right-maximal substrings}: substrings that occur in the text and can be followed by at least two distinct characters. These substrings correspond to the branching points, or internal nodes, of the suffix tree. Suffixient sets capture this branching structure by covering all one-character extensions, also called \emph{right extensions}, of such right-maximal substrings.
Smallest suffixient sets are closely related to \emph{string attractors}~\cite{KempaP18}: a string attractor is a set of positions that intersects an occurrence of every distinct substring, and the size $\gamma$ of the smallest attractor is a fundamental measure of repetitiveness. However, computing a smallest attractor is NP-hard~\cite{KempaP18}. In contrast, the size $\chi$ of a smallest suffixient set can be computed in linear time~\cite{cenzato2024suffixient}, while still giving a meaningful repetitiveness measure. In particular, $\gamma \le \chi$, and $\chi$ is bounded by $2r$~\cite{cenzato2024suffixient}, where $r$ is the number of runs in the Burrows-Wheeler transform (BWT)~\cite{Burrows1994ABL} of the reverse string. 


Interestingly, LRS computation and smallest-suffixient-set construction are closely related.
The value $LRS[i]$ measures the longest repeated suffix ending at the newly read character $T[i]$, and its previous occurrences are witnessed by the right extensions represented by a smallest suffixient set of the reverse of $T[1..i]$. This connection motivates us to study online LRS computation together with the online construction of smallest suffixient sets.
We consider the online setting, where the text $T$ is read from left to right and the algorithm must update its output after each prefix $T[1..i]$. 
Such a setting is natural for incrementally generated texts and rules out the use of future characters or offline scans.
Moreover, worst-case update guarantees are desirable in applications where occasional slow updates, which may occur in amortized algorithms, cannot be tolerated.
The main challenge is therefore to compute $LRS[i]$, or maintain a smallest suffixient set for the current prefix, in compressed working space while guaranteeing worst-case update time.

\subparagraph*{Related work.}
The notion of the longest repeating suffix can be traced back to the suffix-tree construction of McCreight~\cite{McCreight76}. Previous online algorithms for computing the LRS follow two main lines of research.

The first line follows Weiner's algorithm~\cite{Weiner73}, explicitly maintaining an online suffix tree and therefore requiring $O(n)$ words of working space.
Amir et al.~\cite{amir2002online} gave an algorithm that computes each LRS value in $O(\log \sigma)$ amortized time per round, where $\sigma$ is the alphabet size. More recently, Sumiyoshi et al.~\cite{Sumiyoshi2024Spire} and K{\"o}ppl and Kucherov~\cite{KopplK26} de-amortized this approach and obtained $O((\log\log n)^c)$ worst-case time per round, for a constant $c$ depending on the alphabet size. 
The second line maintains an incremental BWT-based index of the online text, together with sampled prefix-array and LCS-array values. This direction was initiated by Okanohara and Sadakane~\cite{OkanoharaS08}, whose algorithm uses succinct working space and computes each LRS value in $O(\log^3 n)$ amortized time. The main bottleneck in their running time is the support for range minimum queries (RMQ) over the dynamic LCS array. Prezza and Rosone~\cite{prezza2020faster} observed that these range minimum queries are not necessary for LRS computation, and improved the time bound to $O(\log^2 n)$ amortized time per round in compressed working space. Thus, compared with the suffix-tree-based line, this line via a BWT-based index is more space-efficient, although previous bounds are amortized.

\begin{table}[htbp]
\centering
\caption{A summary of online algorithms for longest repeating suffix (LRS) and smallest suffixient set (SSS) computation.
Here, $H_k$ denotes the $k$-th-order empirical entropy.
In the bounds of~\cite{Sumiyoshi2024Spire, KopplK26}, the exponent $c$ is a constant that depends on the alphabet size.}
\label{tab:lpf_array_summary}
\smallskip
\small 
\begin{tabularx}{\textwidth}{p{2cm} p{2.2cm} p{3.8cm} p{2cm} Y}
\toprule
\textbf{Problem} & \textbf{Source} & \textbf{Space (bits)} & \textbf{Time} & \textbf{Remarks} \\
\midrule
LRS &\cite{amir2002online} & $O(n\log n)$ & $O(\log \sigma)$ & amortized \\
\addlinespace
LRS &\cite{OkanoharaS08} & $n \log \sigma + o(n \log \sigma)+O(n)$ & $O(\log^3 n)$ & amortized \\ 
\addlinespace
LRS &\cite{prezza2020faster} & $nH_k+o(n\log\sigma)+O(n)+\sigma\log n+o(\sigma\log n)$ & $O( \log^2 n)$ & amortized, $k\in o(\log_\sigma n)$ \\ 
\addlinespace
LRS &\cite{Sumiyoshi2024Spire,KopplK26} & $O(n\log n)$ & $O((\log\log n)^{c})$ & worst case \\ 
\hline
SSS & \cite{fujimaru2026smallestsufficientsetseffectiveness}& $O(n\log n)$ & $O(\log \sigma)$ & amortized \\ 
\addlinespace
SSS &\cite{koppl2026smallestsuffixientsetmaintenance} & $O(n\log n)$ & $O((\log \log n)^2)$ & worst case\\
\hline
Both &This paper & $O(n+r \log n)$  & $O(\frac{\log^2 n}{\log \log n})$ & worst case\\
\addlinespace
Both &This paper & $O(n\log \log n + r\log n)$  & $O((\frac{\log n}{\log\log n})^2)$ & worst case\\
\bottomrule
\end{tabularx}
\end{table}

Cenzato et~al.~\cite{cenzato2024suffixient} gave an $O(n+r\log\sigma)$-time construction of smallest suffixient sets in compressed working space under a streaming model~\cite{Boucher2019PFP,sanaullah2026rlbwt}, later improved to $O(n)$ time by Bonizzoni et~al.~\cite{BonizzoniGR26}.
In this model, the input text is fixed, and the BWT, LCS, and PA arrays are scanned sequentially in a single pass; in contrast, our online setting requires maintaining a smallest suffixient set for every prefix $T[1..i]$, without access to future characters.
In the online regime, Fujimaru et~al.~\cite{fujimaru2026smallestsufficientsetseffectiveness} gave the first algorithm for maintaining a smallest suffixient set, based on Ukkonen's online suffix-tree construction algorithm~\cite{ukkonen1995line}. Their algorithm takes $O(\log \sigma)$ amortized time per round. K\"oppl and Kucherov~\cite{koppl2026smallestsuffixientsetmaintenance} subsequently de-amortized the running time by using Weiner's suffix-tree construction~\cite{Weiner73}. Both algorithms explicitly maintain a suffix tree and therefore require $O(n)$ words of working space. 
A detailed comparison between previous results and ours is given in Table~\ref{tab:lpf_array_summary}.


\subparagraph{Our contributions and technique overview.}
We present the first online algorithm for constructing smallest suffixient sets in compressed working space. 
For online LRS computation in compressed space, we improve the $O(\log^2 n)$ amortized time per character of Prezza and Rosone~\cite{prezza2020faster} to worst-case bounds, obtaining two space-time trade-offs.
Let $r$ denote the number of runs in the BWT of the reverse of the text. Our algorithms achieve the following bounds for both online LRS computation (Theorem~\ref{theorem-LPFA}) and online construction of smallest suffixient sets (Theorem~\ref{theorem-SSS}):
\begin{itemize}
\item $O(r\log n+n)$ bits of working space and $O(\log^2 n/\log\log n)$ worst-case time per round;
\item $O(r\log n+n\log\log n)$ bits of space and $O((\log n/\log\log n)^2)$ worst-case time per round.
\end{itemize}
In particular, for highly repetitive texts satisfying $r=O(n/\log n)$, the first trade-off uses $O(n)$ bits of working space, while the second uses $O(n\log\log n)$ bits.


We also prove in Theorem~\ref{theorem-lower-bound-LRS} that any
deterministic online algorithm for computing LRS requires
\(\Omega(n)\) bits of peak working space in the worst case, even over a
constant-size alphabet. We then establish two reductions from online LRS
computation to online suffixient-set maintenance. The first shows that
the current LRS value can be recovered from any smallest suffixient set
augmented with the length of the supermaximal right extension represented
by each selected position, which we call a \emph{length-annotated SSS}.
The second shows that, when the rightmost occurrence of every
supermaximal right extension is selected, the LRS value can be recovered
from the selected positions alone; we call this canonical set the
\emph{rightmost SSS}. Consequently, the same lower bound applies both to
maintaining an arbitrary length-annotated SSS
(Theorem~\ref{theorem-lower-bound-SSS}(a)) and to maintaining the
position-only rightmost SSS
(Theorem~\ref{theorem-lower-bound-SSS}(b)).
To the best of our knowledge, these are the first explicit working-space lower bounds for online LRS computation and online construction of smallest suffixient sets.
These lower bounds establish that, over constant-size alphabets, the compressed-space bounds of Okanohara and Sadakane and of Prezza and Rosone are optimal up to constant factors. Moreover, when \(r=O(n/\log n)\), our first trade-off uses \(O(n)\) bits, matching the general worst-case lower bound in terms of \(n\).


To obtain these bounds, instead of explicitly maintaining an online
suffix tree as in previous work~\cite{amir2002online,Sumiyoshi2024Spire,
	koppl2026smallestsuffixientsetmaintenance}, we maintain an incremental run-length BWT-based index together with sampled representations of the prefix array ($PA$) and the longest common suffix array ($LCS$). As in previous compressed-space solutions~\cite{OkanoharaS08,
	prezza2020faster}, we partition $PA$ and $LCS$ into blocks of size $O(L)$. The solution of Prezza and Rosone~\cite{prezza2020faster} retrieves an arbitrary $PA$ entry by performing $O(\log n)$ successive $LF/FL$-mapping
steps~\cite{FerraginaM00}, each taking $O(\log n)$ time. We avoid these repeated mappings by
enumerating the entries of a block through $\phi$ queries, where $\phi(PA[j])=PA[j-1]$.
By adapting the reducibility of permuted LCP values
\cite{KarkkainenMP09,GagieNP20} to the prefix-array/LCS setting, each
$\phi$ query and the corresponding $LCS$ value can be recovered from samples associated with the run tops of the BWT. We maintain these samples
in a \emph{dynamic fusion} tree~\cite{DBLP:journals/jcss/FredmanW94, PatrascuT14}, which
supports the required predecessor and successor searches in $O(1+\log_w r)$ time. Consequently, enumerating a given block,
and hence accessing an arbitrary $PA$ or $LCS$ entry, requires $O(L(1+\log_w r))$ time.
Choosing the block size $L$ yields our two space-time trade-offs.

\section{Preliminaries}
\label{sect-prel}


We present all results in the word-RAM model with word size $w=\Theta(\log n)$ bits, where $n$ denotes the final length of the text. To instantiate the two space-time trade-offs, we assume that $n$ is known from the first round, so that the block-size parameter $L$ and the B-tree fan-out $F$ can be fixed accordingly. The general bounds for arbitrary fixed $L$ and $F$ do not otherwise require advance knowledge of $n$.

Let $\Sigma = \{1, \dots, \sigma\}$ be an integer alphabet.
Given a sequence $A[1..n']$ over $\Sigma$, the query $\rank_c(A,j)$ returns the number of occurrences of the symbol $c$ in $A[1..j]$, for $c\in\Sigma$. The query $\select_c(A,j)$ returns the position of the $j$-th occurrence of $c$ in $A$, for $1\le j\le \rank_c(A,n')$; if no such occurrence exists, then $\select_c(A,j)$ returns $n'+1$.

Let $\$ $ be a unique sentinel symbol smaller than every character in $\Sigma$. Let $S[1..m]$ be a string over $\Sigma \cup \{\$\}$ whose first character is $\$$, and assume that $\$$ does not occur elsewhere in $S$. We write $\overleftarrow{S}$ for the reversal of $S$. For a character $c \in \Sigma$, we write $Sc$ or $S\cdot c$ for the string obtained by appending $c$ to the end of $S$.

The prefix array $PA(S)[1..m]$ is the permutation of $\{1,\ldots,m\}$ that orders the prefixes of $S$ colexicographically, that is, $S[1..PA(S)[1]],\ldots,S[1..PA(S)[m]]$ are sorted by the lexicographic order of their reversals. 
See Figure~\ref{fig:exp-lcs-pa-bwt-update} for an example.
We denote by $PA^{-1}(S)[1..m]$ the inverse permutation of $PA(S)$, namely, $PA^{-1}(S)[p]=j$ if and only if $PA(S)[j]=p$. 
The Burrows-Wheeler transform~\cite{Burrows1994ABL} $BWT(\overleftarrow{S})[1..m]$ of $\overleftarrow{S}$ is defined as follows: $BWT(\overleftarrow{S})[j]=S[PA(S)[j]+1]$ if $PA(S)[j]<m$, and $BWT(\overleftarrow{S})[j]=\$$ otherwise. 
A run in $BWT(\overleftarrow{S})$ is a maximal interval $[b..e]$ such that all entries $BWT(\overleftarrow{S})[b..e]$ are equal. For a run $[b..e]$, we call position $b$ the \emph{run top} and position $e$ the \emph{run bottom}.

The \emph{longest common suffix array} $LCS(S)[1..m]$ is defined by $LCS(S)[1]=0$; for $2\le j\le m$, $LCS(S)[j]$ is the length of the longest common suffix of the two adjacent prefixes $S[1..PA(S)[j-1]]$ and $S[1..PA(S)[j]]$.
Following the \emph{permuted longest common prefix} array~\cite[Section~4.2]{KarkkainenMP09}, we define the \emph{permuted longest common suffix array} $PLCS(S)[1..m]$ by $PLCS(S)[j]=LCS(S)[PA^{-1}(S)[j]]$ for every $j\in[1..m]$. 

%
When the string $S$ is clear from the context, we write $PA$, $PA^{-1}$, $LCS$, and $PLCS$ for $PA(S)$, $PA^{-1}(S)$, $LCS(S)$, and $PLCS(S)$, respectively.


\begin{figure}[t]
    \centering
    \includegraphics[width=1\textwidth]{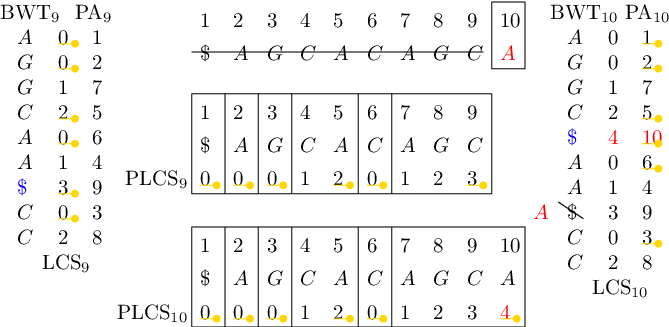}
    \caption{Example of updating the arrays LCS, BWT, PA, and PLCS. The text prefix $T[1..9]=\$AGCACAGC$. Appending the character $A$ yields $T[1..10]=\$AGCACAGCA$. The left and right tables show BWT, PA, and LCS before and after the update, respectively, while the middle tables show the corresponding PLCS values indexed by text position. Yellow dots mark irreducible PLCS positions, equivalently prefix-array entries whose inverse positions are run tops in the corresponding BWT. In the update from round $9$ to round $10$, the new prefix is inserted at co-lexicographic rank $4$, and the affected entries are highlighted in red.}
    \label{fig:exp-lcs-pa-bwt-update}
\end{figure}

\begin{restatable}{proposition}{runbound}\label{pro-bound-on-r}
Let $r_i$ denote the number of runs in $BWT(\overleftarrow{T[1..i]})$ for $i\in[1..n]$.
Then, it follows that $r_i\le r_n+2$, for each $i$.
\end{restatable}

\begin{proof}
Let $R=\overleftarrow{T[1..n]}$ and $R_i=\overleftarrow{T[1..i]}$. Then $R_i$ is a
suffix of $R$. The suffixes of $R_i$ appear in the same relative lexicographic order
as the corresponding suffixes of $R$. Hence $BWT(R_i)$ is obtained from the
subsequence of $BWT(R)$ induced by those suffixes, except that the character
preceding the whole suffix $R_i$ is replaced by the sentinel $\$$. Taking a
subsequence cannot increase the number of runs, and changing one symbol can
increase the number of runs by at most two. Therefore $r_i\le r_n+2$.
\end{proof}



\subparagraph*{Reducible and irreducible values of $PLCS$.}
For a text position $p$, we define $\phi(p)$ to be the prefix-array entry immediately preceding $p$ in co-lexicographic order, that is, if $p=PA[j]$ with $j>1$, then $\phi(p)=PA[j-1]$; otherwise, $\phi(p)=0$ with $p=PA[1]$.
We say that $PLCS[p]$ is \emph{reducible} if $p<|S|$, $\phi(p)<|S|$, and $S[p+1]=S[\phi(p)+1]$. Otherwise, we say that $PLCS[p]$ is \emph{irreducible}, and we also call $p$ an irreducible position. Equivalently, $p$ is irreducible if and only if $PA^{-1}[p]$ is a run top of $BWT(\overleftarrow{S})$.
Figure~\ref{fig:exp-lcs-pa-bwt-update} illustrates the reducible and irreducible $PLCS$ values.

\begin{restatable}{proposition}{irreducibleprop}\label{prop-irreducible}
    Let $p_1,\ldots,p_r$ be the sorted list of values $p$ such that
$PA^{-1}(S)[p]$ is a run top of $BWT(\overleftarrow{S})$, and define
$p_0=0$.
For any $1\le j\le |S|$, let $p_h$ be the successor of $PA[j]$ in this list,
that is, $p_{h-1} < PA[j] \le p_h$.
Then, $LCS[j]=PLCS[PA[j]]=PLCS[p_h] - (p_h-PA[j])$.
Moreover, if $j>1$, then
$\phi(PA[j])=\phi(p_h) - (p_h-PA[j]).$
\end{restatable}

\begin{proof}
    If $PLCS[p]$ is reducible, then the same argument as in~\cite[Lemma~4]{KarkkainenMP09}, adapted to the prefix-array/LCS setting, gives $PLCS[p]=PLCS[p+1]-1$. 
    
    Recall that the irreducible positions are exactly the values $p$ such that $PA^{-1}[p]$ is a run top of $BWT(\overleftarrow{S})$. Hence, by the choice of $p_h$, all positions $PA[j],PA[j]+1,\ldots,p_h-1$ are reducible, while $p_h$ is irreducible. 
Therefore, applying the recurrence to $p=PA[j],PA[j]+1,\ldots,p_h-1$ gives
\[
\begin{array}{rcl}
	PLCS[PA[j]]       &=& PLCS[PA[j]+1]-1,\\
	PLCS[PA[j]+1]     &=& PLCS[PA[j]+2]-1,\\
	& & \vdots\\
	PLCS[p_h-1]       &=& PLCS[p_h]-1.
\end{array}
\]
By telescoping, we obtain $PLCS[PA[j]]=PLCS[p_h]-(p_h-PA[j])$.
    Since $LCS[j]=PLCS[PA[j]]$, we have $LCS[j]=PLCS[PA[j]]=PLCS[p_h]-(p_h-PA[j])$.

    Assume $j>1$, so that $\phi(PA[j])$ is nonzero. For every reducible position $p$, the same argument as above also implies $\phi(p+1)=\phi(p)+1$: the two prefixes $S[1..\phi(p)]$ and $S[1..p]$ are adjacent in co-lexicographic order, and, since $S[\phi(p)+1]=S[p+1]$, their extensions $S[1..\phi(p)+1]$ and $S[1..p+1]$ remain adjacent in the same order. Therefore, repeatedly applying this recurrence over the reducible positions $PA[j],PA[j]+1,\ldots,p_h-1$ gives $\phi(p_h)=\phi(PA[j])+(p_h-PA[j])$. Hence, $\phi(PA[j])=\phi(p_h)-(p_h-PA[j])$.
\end{proof}

	\begin{example}\label{app-example-irreducible}
	Consider the text $T[1..10]=\$AGCACAGCA$ shown in Figure~\ref{fig:exp-lcs-pa-bwt-update}. The run tops of $BWT(\overleftarrow{T[1..10]})$ occur at prefix-array positions $1,2,4,5,6,9$. Hence the corresponding prefix-array entries are $1,2,5,10,6,3$, and the sorted list of irreducible positions is $p_1,\ldots,p_r=1,2,3,5,6,10$. Now take $j=3$. Then $PA[j]=7$, and the successor of $PA[j]$ among the irreducible positions is $p_h=10$. Proposition~\ref{prop-irreducible} gives $LCS[3]=PLCS[10]-(10-7)=4-3=1$, which agrees with the value shown in the figure.  The same example also illustrates the formula for $\phi$: since $PA[3]=7$, we have $\phi(PA[3])=PA[2]=2$; on the other hand, $\phi(10)=5$, and therefore $\phi(PA[3])=\phi(10)-(10-7)=5-3=2$.
    	\end{example}

\subparagraph*{Defining longest repeating suffixes and suffixient sets.}\phantomsection\label{def-lpfa} The longest repeating suffix array, denoted by \(LRS[1..n]\), is defined as follows. For each position \(1\le i\le n\), \(LRS[i]\) is the length of the longest suffix of \(T[1..i]\) that occurs previously in \(T\), ending at some position \(j<i\). Equivalently, \(LRS[i]=\max\{\ell\ge 0 : \text{there exists } j<i \text{ such that } T[i-\ell+1..i]=T[j-\ell+1..j]\}\).
As shown in Figure~\ref{fig:exp-lcs-pa-bwt-update}, $LRS[10]=4$, since $AGCA$ is the longest suffix of $T[1..10]$ that appears twice and its length is $4$.

Let \(S\) be a string and let \(\mathcal{F}_S\) denote the set of all substrings of \(S\).
A substring \(u\in\mathcal{F}_S\) is \emph{right-maximal} if there exist two distinct
characters \(a,b\in\Sigma \cup \{\$\}\) such that \(ua,ub\in\mathcal{F}_S\). For a right-maximal
substring \(u\), a pair \((u,a)\) is called a \emph{right extension} if \(ua\in\mathcal{F}_S\).
Let \(\mathcal{E}_r(S)\) denote the set of all right extensions of \(S\).
A right extension $(u,a)\in\mathcal{E}_r(S)$ is called a \emph{supermaximal right extension}, or an \emph{SRE}, if there is no character $c\in\Sigma$ such that $(cu,a)$ is also a right extension. Let $\mathcal{S}_r(S)$ denote the set of all SREs of $S$.
A \emph{smallest suffixient set} of $S$ is a set of positions $P\subseteq \{1,\ldots,|S|\}$ obtained by choosing, for every SRE $(u,a)\in\mathcal{S}_r(S)$, one occurrence of $ua$ in $S$ and storing its ending position in $P$. Equivalently,
\(P\) is a smallest set of positions such that, for every right extension
\((u,a)\in\mathcal{E}_r(S)\), there exists an occurrence of \(ua\) in \(S\) ending at a
position of \(P\). 
Let $\chi(S)=|\mathcal{S}_r(S)|$ denote the size of a smallest suffixient set. Then $\chi(S)$ is at most twice the number of runs in $BWT(\overleftarrow{S})$~\cite[Lemma~10]{cenzato2024suffixient}.

For each SRE \((u,a)\in\mathcal{S}_r(S)\), let $\rho(u,a)=\max\{j : S[j-|ua|+1..j]=ua\}$ denote the ending position of the rightmost occurrence of $ua$ in $S$. A \emph{canonical rightmost-occurrence smallest suffixient set} of \(S\), or simply the \emph{rightmost SSS} (\emph{RSSS}), is ${RSSS}(S)=\{\rho(u,a) : (u,a)\in\mathcal{S}_r(S)\}$. It is uniquely determined by \(S\) and is therefore canonical.
A \emph{length-annotated smallest suffixient set}, or
\emph{length-annotated SSS}, of \(S\) is a set of records $\widehat{P}=\{\langle j,|ua|\rangle : (u,a)\in\mathcal{S}_r(S)\}$, where, for each SRE \((u,a)\), \(j\) is the ending position of a chosen occurrence of \(ua\) in \(S\). Thus, the projection of \(\widehat{P}\) onto its first components is a smallest suffixient set of \(S\). The \emph{length-annotated rightmost SSS} of \(S\) is $\widehat{RSSS}(S)=\{\langle \rho(u,a),|ua|\rangle : (u,a)\in\mathcal{S}_r(S)\}$.
It augments every position in $RSSS(S)$ with the length of the corresponding SRE.


\subparagraph*{Online updates of the run-length BWT, PA, and LCS.}

Before concluding the preliminaries, we recall the standard online update rules for the $BWT$, the prefix array, and the $LCS$ array under the co-lexicographic ordering of prefixes~\cite{OkanoharaS08}. In our work, the $BWT$ relevant operations are supported by a dynamic run-length BWT representation~\cite[Theorem~2]{DBLP:journals/algorithmica/PolicritiP18}.

\begin{lemma}[{\cite[Theorem~2]{DBLP:journals/algorithmica/PolicritiP18}}]
\label{lem-rlBWT-online}
Let $r_i$ denote the number of runs in $BWT(\overleftarrow{T[1..i]})$.
There exists a data structure using $O(r_i)$ words that represents
$BWT(\overleftarrow{T[1..i]})$ and supports each of the following operations
in $O(\log r_i)$ time:
\begin{itemize}
    \item $\access$, $\rank$ and $\select$ queries on $BWT(\overleftarrow{T[1..i]})$;
    \item computing \footnote{This part is implemented by an auxiliary dynamic cumulative-count structure using $O(\sigma_i) \subseteq O(r_i)$, where $\sigma_i$ is the alphabet size for $T[1..i]$.}, for any $c\in\Sigma$, the number $C[c]$ of occurrences in
    $BWT(\overleftarrow{T[1..i]})$ of symbols that are smaller than $c$;
    \item inserting a symbol into $BWT(\overleftarrow{T[1..i]})$; and
    \item replacing the unique terminator $\$$ by a new symbol in $BWT(\overleftarrow{T[1..i]})$.
\end{itemize}
\end{lemma}

Let $S$ be the current string and let $c$ be the next character.  We describe how to transform $BWT(\overleftarrow{S})$, $PA(S)$, and $LCS(S)$ into $BWT(c\overleftarrow{S})$, $PA(Sc)$, and $LCS(Sc)$.

\phantomsection\label{def-online-update-notation}

Let $\ell$ denote the co-lexicographic rank of the new prefix $Sc$
among all prefixes of $Sc$, or equivalently, the insertion position of
$Sc$ among the old prefixes of $S$, and let $s$ denote the position of $\$$ in the old $BWT(\overleftarrow{S})$.
Among them, $s$ can be obtained by a $\select_{\$}$ query, and $\ell$ is essentially $C[c] + \rank_c(BWT(\overleftarrow{S}),s) + 1$, where $C[c]$ is the number of symbols in $BWT(\overleftarrow{S})$ that are smaller than $c$.
By Lemma~\ref{lem-rlBWT-online}, both $s$ and $\ell$ can be computed in $O(\log r_i)$ time.

To obtain $BWT(c\overleftarrow{S})$, replace the unique symbol $\$$ at the position $s$ with $c$, and insert a new symbol $\$$ at
position $\ell$.
To obtain $PA(Sc)$, insert $|Sc|$ at position $\ell$ of $PA(S)$.
Figure~\ref{fig:exp-lcs-pa-bwt-update} shows an example on the updates.

It remains to update the $LCS$ array. 
Let $b<s$ denote the largest position in the old $BWT(\overleftarrow{S})$ where the symbol $c$ appears, 
and let $f > s$ denote the smallest position in the old $BWT(\overleftarrow{S})$ where the symbol $c$ appears. 
More precisely, let $q=\rank_c(BWT(\overleftarrow{S}),s)$.
Then, $b=\select_c(BWT(\overleftarrow{S}),q)$ if $q>0$; or $b$ does not exist, otherwise.
Similarly, if $q<\rank_c(BWT(\overleftarrow{S}),|S|)$, then $f=\select_c(BWT(\overleftarrow{S}),q+1)$; otherwise, $f$ does not exist.

Following~\cite[Algorithm 1]{OkanoharaS08}, if $b$ exists, then insert $1+\min\{LCS[b+1], LCS[b+2], \dots, LCS[s]\}$ at position $\ell$ of $LCS(S)$; otherwise, insert 0 at position $\ell$.
If $\ell+1\le |Sc|$, then do the following. If $f$ exists, update the entry
at position $\ell+1$ of $LCS(Sc)$ to $1+\min\{LCS[s+1], LCS[s+2], \dots, LCS[f]\}$;
otherwise, update that entry to $0$.


By Lemma~\ref{lem-rlBWT-online}, the run-length BWT of $\overleftarrow{T[1..i]}$ can be maintained online in $O(r_i)$ words of working space. Moreover, the rank $\ell$ and the positions $s$, $b$, and $f$ can be computed by cumulative-count, rank, and select queries on the old $BWT(\overleftarrow{T[1..i-1]})$ in $O(\log r_{i-1})$ time. Since the BWT update from round $i-1$ to round $i$ changes the number of runs by at most a constant, this is $O(\log r_i)$ time.
Hence, both the BWT update and the computation of the insertion position for the new $PA$ entry take $O(\log r_i)$ time.

To compute the new LCS entries at position $\ell$ and $\ell+1$, range minimum queries over the old LCS array are required.
In Section~\ref{sect-ds-rmq}, we give the range-minimum-query data structure used in our algorithm, replacing the dynamic RMQ component that leads to the $O(\log^3 n)$-time bound of Okanohara and Sadakane~\cite{OkanoharaS08}.



\section{Online Data Structures for $\phi$ Queries and Access to $PA$ and $LCS$}
\label{sect-random-access}

In this section, we first introduce an online data structure that, at each
round $i$, maintains the prefix-array entries occurring at the run tops of
$BWT(\overleftarrow{T[1..i]})$, together with their $\phi$ values and their
$LCS$ values. We then explain how to use this data structure to support $\phi$ and
$LCS$ access at each round. 
The data-structure updates follow in the subsequent subsection.

\subsection{The data structures and queries}
\label{sect-ds-basic}

\phantomsection\label{def-ds-run-tops}
Consider any round $i \ge 2$. Let $t_1,\ldots,t_{r_i}$ denote the run tops of
$BWT(\overleftarrow{T[1..i]})$, where $r_i$ is the number of runs in
$BWT(\overleftarrow{T[1..i]})$.
We maintain the set $\{\, PA[t_x] : 1\le x\le r_i \,\}$ using a dynamic \emph{fusion tree}~\cite{DBLP:journals/jcss/FredmanW94, PatrascuT14}. For each key $PA[t_x]$, we store
$\phi(PA[t_x])$ and $LCS[t_x]$ as satellite data.
The fusion tree uses $O(r_i)$ words and supports insertions, deletions,
predecessor queries, and successor queries on this set in
$O(1+\log_w r_i)$ time.
By Proposition~\ref{prop-irreducible}, given any prefix-array entry $PA[j]$, we can
use the fusion tree to compute $LCS[j]$ and $\phi(PA[j])$ in
$O(1+\log_w r_i)$ time.

\phantomsection\label{def-block-decomposition}
Let $L$ be the block parameter to be decided later. 
Inspired by the work of Okanohara and Sadakane~\cite{OkanoharaS08},
at each round $i$ we conceptually partition the entries of $PA(T[1..i])$
into blocks of size between $L$ and $2L$. 
If $i\le L$, then all the entries are assigned to the same block.
We build a B-tree $\mathcal{T}$ with fan-out $F$ over these blocks, so that each leaf corresponds to one block. 
A B-tree~\cite{bayer1970organization} is a rooted ordered tree in which all leaves have the same depth, every non-root internal node has between $\lceil F/2\rceil$ and $F$ children, and the root has between $2$ and $F$ children unless it is the only node.
Hence, $\mathcal{T}$ has $O(1+i/L)$ leaves and height $O(1+\log_F(1+i/L))$.
The entries of a block are not stored explicitly at the corresponding leaf. Instead, each leaf stores the size of its block, called its weight, together with the last $PA$-entry of the block.
At each internal node $v$, we store, as its weight, the sum of the weights of the descendant leaves of $v$.

Since there are $\Theta(1+i/L)$ blocks, the B-tree contains
$m=\Theta(1+i/L)$ leaves and $O(m/F)$ internal nodes.
Each internal node has $O(F)$ child pointers of $w$ bits each.
Storing the tree requires $O(w\cdot (m+ F\frac{m}{F}))=O(w(1+i/L))$ bits.
The auxiliary data stored at each node costs $O(w)$ bits.
Hence, the total space cost of the data structures built over the B-tree is bounded by $O(w\cdot m)=O(w(1+i/L))$ bits.


\phantomsection\label{def-btree-navigation}
Given a rank $j\in[1..i]$, we start at the root of $\mathcal{T}$. At the current internal node, we scan its children from left to right and find the child such that one of its leaf descendant corresponding to the block that contains the entry $PA[j]$.
More precisely, let $s$ be the total weight of the children scanned before the chosen child, and let $z$ be the weight of the chosen child. We choose the unique child satisfying $s<j\le s+z$, and then continue the search recursively at this child with local rank $j-s$. When a leaf is reached, the corresponding block is exactly the block containing the entry $PA[j]$. 
At each level, we inspect at most $F$ children. Since the height of $\mathcal{T}$ is $O(1+\log_F(1+i/L))$, the block containing $PA[j]$ can be found in $O(F(1+\log_F(1+i/L)))$ time.

The data structures above allow us to enumerate the $PA$ entries and the
corresponding $LCS$ entries in any block. Given an arbitrary position $j\in [1..i]$, we first find the block that contains the entry $PA[j]$ as described earlier, taking $O(F(1+\log_F(1+i/L)))$ time.
Then, starting from
the last $PA$-entry stored for that block, we iteratively apply $\phi$
queries $k-1$ times, where $k$ is the size of the block. Since
$k=O(L)$ and each $\phi$ query is supported using the fusion tree in
$O(1+\log_w r_i)$ time, all $PA$ entries in the block can be enumerated in $O(F(1+\log_F(1+i/L))+L\cdot(1+\log_w r_i))$ time. During the same enumeration, we recover the corresponding $LCS$ value of each enumerated $PA$ entry using Proposition~\ref{prop-irreducible}.

Since each block has size at most $2L$, and each $\phi$ or $LCS$ query takes $O(1+\log_w r_i)$ time, all $PA$ entries in the block, together with their corresponding $LCS$ values, can be enumerated in $O(F(1+\log_F(1+i/L)) + L\cdot(1+\log_w r_i))$ time.
This proves Lemma~\ref{lem-block-PA-LCS}.

\begin{lemma}\label{lem-block-PA-LCS}
In each round $i$, one can maintain a data structure using
$O(w(r_i+i/L))$ bits of space such that, given a rank
$j\in[1..i]$, it can enumerate all $PA$ entries and the corresponding
$LCS$ values in the block containing $PA[j]$ in
$O(F(1+\log_F(1+i/L))+L\cdot(1+\log_w r_i))$
time. Consequently, retrieving $PA[j]$, $\phi(PA[j])$, and $LCS[j]$ for
an arbitrary rank $1\le j\le i$ takes the same time.
\end{lemma}

For readability, throughout the rest of the paper we write $H_i := 1+\log_F(1+i/L)$ and $\BlockTime{i} := F\cdot H_i + L\cdot(1+\log_w r_i)$. Thus, the B-tree has height $O(H_i)$ in round $i$, while $O(\BlockTime{i})$ is the cost of locating a relevant block and enumerating its $O(L)$ entries using $\phi$ queries. In particular, by Lemma~\ref{lem-block-PA-LCS}, block enumeration and random access to $PA[j]$, $\phi(PA[j])$, and $LCS[j]$ take $O(\BlockTime{i})$ time.

%

\subsection{Updating the B-tree and the fusion tree}
\label{sect-update-B-tree-Fusion-Tree}

In this section, we detail the updates in each round (specifically, from round $i-1$ to round $i$) on the maintained data structures, including the B-tree and the dynamic fusion tree.

\subsubsection{Updating the B-tree.}
\label{sect-update-balanced}

 After the insertion of the $PA$-entry, if the block that contains the inserted $PA$-entry does not overflow, that is, its size is at most $2L$, we update the weight of the leaf corresponding to the block by increasing its original weight by one and, if necessary, the last $PA$-entry stored for this block.
The weight of every ancestor of this leaf increases by one, and we update these weights along the root-to-leaf path.
This involves $O(\BTHeight{i})$ nodes and takes $O(\BTHeight{i})$ time.

If the block overflows, then we split it into two blocks of sizes $L$ and $L+1$, respectively. 
Before performing the split, we first
enumerate the $PA$-entries in the old block using
Lemma~\ref{lem-block-PA-LCS}, which takes $O(\BlockTime{i})$ time. 
We then insert the newly added $PA$-entry into this enumerated sequence
at its proper position. When this
sequence is split into two blocks of size $\Theta(L)$, we can identify
and store the last $PA$-entry of each resulting block.

The split of the block also causes its corresponding leaf to split into two leaves.
We replace the old leaf corresponding to the old block by the two new leaves, so both new leaves share the same common ancestors.
We also update the auxiliary data of each new leaf, including its weight and the last $PA$-entry in its corresponding block, taking constant time.

If the common parent $v$ of both leaves still has at most $F$ children after the insertion, then no further node split is required. 
We update the weight of each ancestor as before, which takes $O(\BTHeight{i})$ time.

Otherwise, $v$ has $F+1$ children and overflows. In this case, we split $v$ into two consecutive nodes, each containing $\Theta(F)$ children, and compute the weights of both new nodes, taking $\Theta(F)$ time.
This split may in turn cause the parent of both new nodes to overflow, and the same splitting procedure is then applied recursively towards the root. If the root overflows, it is split and a new root is created.
Thus, the update costs $O(F)$ time at each level, for a total of $O(F\cdot\BTHeight{i})$ time. 

The preceding analysis covers both the non-overflow and overflow cases and immediately yields Lemma~\ref{lem-update-B-tree}.

\begin{lemma}\label{lem-update-B-tree}
In each round $i$, updating the B-tree and the auxiliary data stored at its nodes requires
$O(\BlockTime{i})$
time in the worst case.
\end{lemma}

\subsubsection{Updating the dynamic fusion tree.}
\label{sect-fusion-tree-update}

As shown in the preliminaries, each round modifies the BWT by
replacing the unique occurrence of $\$$ by the newly read character $T[i]$ and by
inserting a new occurrence of $\$$. These two operations may change the set of
run tops. We now describe how to update the dynamic fusion tree accordingly.

\phantomsection\label{def-fusion-update-notation}
Let $B_{old}=BWT(\overleftarrow{T[1..i-1]})$
and let $PA_{old}=PA(T[1..i-1])$ and $LCS_{old}=LCS(T[1..i-1])$.
Let $c=T[i]$, let $s$ be the position of the unique occurrence of $\$$ in
$B_{old}$, and let $\ell$ be the co-lexicographic rank of $T[1..i]$ among all prefixes of
$T[1..i]$. By Lemma~\ref{lem-rlBWT-online}, both the position $s$ and the rank $\ell$ can be computed in
$O(\log r_{i})$ time. 


Before modifying the maintained structures, we retrieve all old prefix-array
entries that may be needed during the update. In particular, we retrieve
$PA_{old}[s]$, $PA_{old}[s+1]$ if $s<i-1$, $PA_{old}[\ell]$ if
$\ell<i$, and $PA_{old}[i-1]$ if $\ell=i$. 
By Lemma~\ref{lem-block-PA-LCS}, these $O(1)$ entries can be retrieved in $O(\BlockTime{i})$ time.
We also compute the value $p$ that will become $\phi(i)$ after inserting the new prefix: set $p \leftarrow 0$ if $\ell=1$, $p \leftarrow \phi(PA_{\mathrm{old}}[\ell])$ if $1<\ell<i$, and $p \leftarrow PA_{\mathrm{old}}[i-1]$ if $\ell=i$.  Using the old dynamic fusion tree and Lemma~\ref{lem-block-PA-LCS}, the value $p$ can be computed in $O(\BlockTime{i})$ time.


Next, we compute the two LCS values affected by the insertion. 
For this purpose, we reuse the notation introduced immediately after Lemma~\ref{lem-rlBWT-online}.
Let $b<s$ be
the largest position in $B_{old}$ containing the symbol $c$, if such a
position exists, and let $f>s$ be the smallest position in $B_{old}$ containing
the symbol $c$, if such a position exists. As shown there, these positions can be found using
rank and select queries on $B_{old}$ in $O(\log r_{i})$ time by Lemma~\ref{lem-rlBWT-online}.

\phantomsection\label{def-fusion-update-lcs-values}
Define $x=1+\min LCS_{old}[b+1..s]$ if $b$ exists and $x=0$ otherwise; similarly, define $y=1+\min LCS_{old}[s+1..f]$ if $f$ exists and $y=0$ otherwise.
The value $x$ is the LCS value inserted at position $\ell$, and $y$ is the
new LCS value at position $\ell+1$ if $\ell<i$. 
Let $Q_{\mathrm{RMQ}}(i)$ denote the worst-case time to answer a
range minimum query over $LCS(T[1..i])$ using the data structure maintained
at the end of round $i$.
As $x$ and $y$ are defined by range minimum queries over $LCS_{\mathrm{old}}=LCS(T[1..i-1])$, they can be computed in $O(Q_{\mathrm{RMQ}}(i-1))$ time.

We now update the dynamic fusion tree. First consider the replacement of
$B_{old}[s]=\$$ by $c$. In the old BWT, position $s$ is a run top, and
position $s+1$ is also a run top if $s<i-1$. After replacing $\$$ by $c$,
some of these run tops may disappear. If $s>1$ and $B_{old}[s-1]=c$, then
position $s$ is no longer a run top, so we delete the key $PA_{old}[s]$ from
the dynamic fusion tree. Similarly, if $s<i-1$ and $B_{old}[s+1]=c$, then
position $s+1$ is no longer a run top, so we delete the key $PA_{old}[s+1]$ from the dynamic fusion tree.
Whenever a key is deleted from the fusion tree, its satellite data is deleted as well.

We then insert the new occurrence of $\$$ at position $\ell$. Since $\$$ is
unique and lexicographically smaller than all other symbols, this creates a
new run top at position $\ell$. The corresponding prefix-array entry is the
new value $i$. Hence, we insert the key $i$ into the dynamic fusion tree, with
satellite data $\phi(i)=p$ and $LCS[\ell]=x$.
If $\ell<i$, then the old entry $PA_{old}[\ell]$ is shifted to position
$\ell+1$. Since the symbol at position $\ell$ is now the unique $\$$, position
$\ell+1$ is also a run top. Therefore, we insert or update the key
$PA_{old}[\ell]$ in the dynamic fusion tree, setting its satellite data to $\phi(PA_{old}[\ell])=i$ and $LCS[\ell+1]=y$.

After these operations, the dynamic fusion tree stores exactly the
prefix-array entries corresponding to the run tops of
$BWT(\overleftarrow{T[1..i]})$, together with their updated $\phi$ and $LCS$
values. The correctness follows from the fact that replacing the old $\$$ by
$c$ can only affect the run-top status of positions $s$ and $s+1$, while
inserting the new $\$$ at position $\ell$ creates run tops at positions
$\ell$ and $\ell+1$ (if $\ell<i$).
Each insertion, deletion, or satellite-data update in the dynamic fusion tree
takes $O(1+\log_w \max\{r_i,r_{i-1}\})$ time. Since only $O(1)$ such operations are performed and
$r_i=\Theta(r_{i-1})$, the total time spent on fusion tree updates is $O(1+\log_w r_i)$.


Altogether, the update of the dynamic fusion tree from round $i-1$ to round $i$ takes $O(\BlockTime{i}+Q_{\mathrm{RMQ}}(i-1)+\log r_i)$ time.
This proves Lemma~\ref{lem-update-fusion}.


\begin{lemma}\label{lem-update-fusion}
For every round \(i>1\), the dynamic fusion tree can be updated in
\(O(\BlockTime{i}+Q_{\mathrm{RMQ}}(i-1)+\log r_i)\) time if an RMQ over \(LCS(T[1..i-1])\) takes \(Q_{\mathrm{RMQ}}(i-1)\) time.
\end{lemma}


In the next section, we present a range-minimum-query solution and bound $Q_{\mathrm{RMQ}}(i-1)$.

\section{Faster Range Minimum Queries over LCS and Prefix Arrays}
\label{sect-ds-rmq}

For range minimum queries over the LCS and prefix arrays, we reuse the B-tree from Section~\ref{sect-ds-basic}, whose leaves represent prefix-array blocks.


\subsection{The data structure}
Recall that we conceptually divide \(PA(T[1..i])\) into blocks of size $\Theta(L)$, and that each leaf of the B-tree $\mathcal{T}$
corresponds to one such block. We divide \(LCS(T[1..i])\) in exactly the
same way: the entries of \(LCS(T[1..i])\) are partitioned into blocks
according to the same block boundaries as the entries of \(PA(T[1..i])\).
Thus, each leaf of $\mathcal{T}$ corresponds both to a block
of \(PA(T[1..i])\) and to the corresponding block of \(LCS(T[1..i])\).

We augment $\mathcal{T}$ with range-minimum information. For
each leaf, we store the minimum LCS and PA values in the corresponding LCS and PA blocks, respectively.
For each internal node, we store the minimum of the values stored at its
children; equivalently, this is the minimum LCS (resp. PA) value among all LCS (resp. PA)
entries in the blocks covered by the subtree rooted at that node. Each
stored minimum uses \(O(w)\) bits. Since the tree has
\(O(1+i/L)\) nodes, this additional aggregate information uses
\(O(w\cdot(1+i/L))\) bits of space.

\subsection{The algorithm for range minimum queries}
We show the algorithm for range minimum queries over the LCS array, while the one over the PA array can be implemented similarly.

Consider a range minimum query over \(LCS(T[1..i])\) with query range
\([e,f]\). We first find the blocks \(b_e\) and \(b_f\) containing
positions \(e\) and \(f\), respectively, by descending $\mathcal{T}$ using the weight information stored at each node. If \(b_e=b_f\), then the
whole query range lies inside one block. In this case, we apply
Lemma~\ref{lem-block-PA-LCS} to enumerate the LCS values in this block
and return the minimum value among the positions \(e,\ldots,f\).

It remains to consider the case \(b_e<b_f\). We decompose the query range
into three parts: \([e,\mathrm{end}(b_e)]\), the complete blocks
\(b_e+1,\ldots,b_f-1\), and \([\mathrm{start}(b_f),f]\).
The first and third parts are contained in the two boundary blocks. We
compute their minima by applying Lemma~\ref{lem-block-PA-LCS} to enumerate
the LCS values in the corresponding boundary blocks and by taking the
minimum over the relevant positions only.

For the middle part, consisting of the complete blocks
$b_e+1,\ldots,b_f-1$, we use the aggregate values stored in the
B-tree $\mathcal{T}$. If $b_e+1>b_f-1$, then the middle part is empty.
Otherwise, let $\lambda_e$ and $\lambda_f$ be the leaves representing
the boundary blocks $b_e$ and $b_f$, respectively, and let $x$ be
their lowest common ancestor in $\mathcal{T}$.

We first traverse the path from $\lambda_e$ upward toward $x$, stopping at the child of $x$ on this path. For every visited node $v$, let $u$ be the parent of $v$. We add to a set $\mathcal{L}$ all children of $u$ that appear after $v$ in the left-to-right order and whose subtrees are strictly before $\lambda_f$.
Symmetrically, we traverse the path from $\lambda_f$ upward toward $x$, stopping at the child of $x$ on this path. For every visited node $v$, let $u$ be the parent of $v$. We add to a set $\mathcal{R}$ all children of $u$ that appear before $v$ in the left-to-right order and whose subtrees are strictly after $\lambda_e$.

The nodes in $\mathcal{L}\cup\mathcal{R}$ form the standard
\emph{canonical decomposition} of the leaf interval strictly between
$\lambda_e$ and $\lambda_f$, as used in \emph{segment trees} and \emph{range trees}~\cite{BergCKO08}.
An illustration is provided in Figure~\ref{fig:B-tree}.

In particular, they root pairwise disjoint subtrees, and the union of the blocks covered by these subtrees is exactly the block interval $b_e+1,\ldots,b_f-1$. 
Since the B-tree $\mathcal{T}$ has height
$\BTHeight{i}$, and at most $O(F)$ nodes are added at each level, we have $|\mathcal{L}\cup\mathcal{R}|=O(F\cdot\BTHeight{i})$. 
Therefore, by taking the minimum of the aggregate LCS minima stored at the nodes in
$\mathcal{L}\cup\mathcal{R}$, we obtain the minimum LCS value in the middle part.

\begin{figure}[t]
    \centering
    \includegraphics[width=1\textwidth]{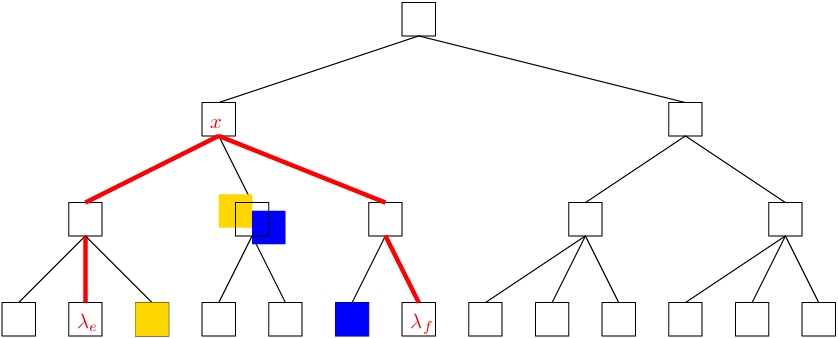}
    \caption{Decomposition of the complete blocks between the boundary blocks $b_e$ and $b_f$. The leaves $\lambda_e$ and $\lambda_f$ represent $b_e$ and $b_f$, respectively, and $x$ is their lowest common ancestor. During the traversal from $\lambda_e$ toward $x$, the yellow nodes are added to $\mathcal{L}$: each is a right sibling of a node on the traversal path, and its subtree lies strictly before $\lambda_f$. Symmetrically, the blue node is collected during the traversal from $\lambda_f$ toward $x$. The subtrees rooted at the selected nodes cover exactly the complete blocks $b_e+1,\ldots,b_f-1$.}
    \label{fig:B-tree}
\end{figure}

Finally, we return the minimum of the values obtained from the left
boundary block, the middle complete blocks, and the right boundary block.

\subparagraph*{Time complexity analysis.}
The two boundary blocks \(b_e\) and
\(b_f\) can be found by root-to-leaf searches in the B-tree,
taking $O(F\cdot\BTHeight{i})$ time. By Lemma~\ref{lem-block-PA-LCS}, the two
boundary minima can be computed in
$O(\BlockTime{i})$ time.
For the middle part, we first find the lowest common ancestor of the
leaves \(\lambda_e\) and \(\lambda_f\). We then traverse the two paths
from these leaves toward the lowest common ancestor $x$ and collect
$O(F\cdot\BTHeight{i})$ nodes from the tree, as described above. For each
collected node, the aggregate minimum stored at that node can be accessed
in constant time. Hence, the middle part costs $O(F\cdot\BTHeight{i})$ time.
Therefore, the total query time
is bounded by $O(F\cdot\BTHeight{i}+\BlockTime{i})\subseteq O(\BlockTime{i})$.


\subsection{Updating the range-minimum-query data structure}
We describe how to maintain the aggregate minimum values from round
$i-1$ to round $i$.
Let $\ell$ be the position at which the new entries are inserted into the
$LCS$ and $PA$ arrays. Recall from the preliminaries that the
update to the $LCS$ array consists of inserting a new value $x$ at position
$\ell$ and, if $\ell<i$, updating the entry shifted to position $\ell+1$
to a new value $y$. The update to the $PA$ array simply inserts the new
value $i$ at position $\ell$.

Let $B$ be the old block into which the new entries are inserted; if
$\ell=i$, then $B$ is the last block. Before modifying $B$, we enumerate
all $PA$-entries and the corresponding $LCS$-values in $B$, using
Lemma~\ref{lem-block-PA-LCS}. This takes
$O(\BlockTime{i})$ time. We then insert the
new $PA$-entry $i$ and the new $LCS$-value $x$ into the enumerated
sequences at the appropriate position. If $\ell<i$, we also replace the
$LCS$-value of the shifted entry at position $\ell+1$ by $y$. Thus, after
this step, we know all $PA$-entries and $LCS$-values in the affected block
after the insertion and update.

If the affected block does not overflow, we recompute the minimum
$PA$-value and the minimum $LCS$-value of this block directly from the
updated sequences. We then update the aggregate minimum values along the
root-to-leaf path of the corresponding leaf, using the same procedure as
for updating aggregate information in Section~\ref{sect-update-balanced}.
This costs $O(F\cdot\BTHeight{i})$ additional time.

If the affected block overflows, then it is split into two blocks of size
$\Theta(L)$, as described in Section~\ref{sect-update-balanced}. Since the
updated sequences are already known explicitly, we can compute the minimum
$PA$-value and the minimum $LCS$-value of each resulting block directly in
$O(L)$ time. We then update the B-tree $\mathcal{T}$, including the aggregate
minimum values and, if necessary, the tree structure, using the same
overflow-handling procedure as in Section~\ref{sect-update-balanced}. This
costs $O(F\cdot\BTHeight{i})$ additional time.

Therefore, the total time for updating the RMQ aggregate information in
round $i$ is
$O(F\cdot\BTHeight{i}+\BlockTime{i})\subseteq O(\BlockTime{i})$. 

Lemma~\ref{lem-RMQ} summarizes the space and time complexity of the range-minimum-query data structure. 

\begin{lemma}\label{lem-RMQ}
In each round \(i\), by augmenting the data structure of
Lemma~\ref{lem-block-PA-LCS} with \(O(w(1+i/L))\) bits, one can support range minimum queries over
\(LCS(T[1..i])\) and \(PA(T[1..i])\) in \(O(\BlockTime{i})\) time.
Given the insertion position \(\ell\) and the new affected \(LCS\) values
\(x\) and \(y\), the RMQ aggregate information can also be updated in
\(O(\BlockTime{i})\) time.
\end{lemma}

\subsection{The next/previous smaller value queries}\phantomsection\label{def-nsv-psv}
Let $LCS_i$ denote $LCS(T[1..i])$. Given a position $j\in [1..i]$ and a
value $\mathit{val}\in [1..n]$, we define $NSV_i(j,\mathit{val})$ as the
smallest index $j'\in [j..i]$ such that $LCS_i[j']<\mathit{val}$, if such
an index exists; otherwise, the query returns $\bot$. Symmetrically,
$PSV_i(j,\mathit{val})$ returns the largest index $j'\in [1..j]$ such that
$LCS_i[j']<\mathit{val}$, or $\bot$ if no such index exists. In
Lemma~\ref{lem-nsv}, we show how to support these two queries using the
data structure introduced above.

\begin{lemma}\label{lem-nsv}
Given a position $j\in[1..i]$ and a value $val\in[1..n]$, the
data structure for range minimum queries over $LCS(T[1..i])$ supports
$NSV_i(j,val)$ and $PSV_i(j,val)$ in
$O(\BlockTime{i})$ time.
\end{lemma}

\begin{proof}
We describe the query algorithm for $NSV_i(j,\mathit{val})$; the algorithm
for $PSV_i(j,\mathit{val})$ is symmetric.

Let $B_1$ be the block containing position $j$, and let $\lambda_1$ be the
leaf representing $B_1$ in the B-tree. We first enumerate the
$LCS_i$-values in $B_1$, starting from position $j$ and continuing to the
end of $B_1$. If during this enumeration we find an entry smaller than
$\mathit{val}$, then the first such entry is exactly
$NSV_i(j,\mathit{val})$, and we return its position.

Suppose that no such entry is found in $B_1$. We then traverse the path, in the B-tree,
from $\lambda_1$ upward toward the root. For each visited node $u$, we
inspect the right siblings of $u$, in left-to-right order, and stop as soon
as we find a sibling whose aggregate minimum is smaller than
$\mathit{val}$. If no such sibling is found during the whole upward
traversal, then no entry smaller than $\mathit{val}$ occurs to the right of
$B_1$, and the query returns $\bot$.

Otherwise, let $v$ be the first sibling found in this way. We descend from
$v$ to a leaf as follows. At each internal node, we inspect its children
from left to right and move to the first child whose aggregate minimum is smaller than $\mathit{val}$. Since the aggregate minimum of the current subtree is always smaller than $\mathit{val}$, such a child always exists.
This descending traversal ends at the leftmost block $B_2$ to the right of
$B_1$ whose minimum $LCS_i$-value is smaller than $\mathit{val}$.

Finally, we enumerate the $LCS_i$-values in $B_2$ from the beginning of the
block and return the first position whose value is smaller than
$\mathit{val}$. By the choice of $B_2$, this position is precisely
$NSV_i(j,\mathit{val})$.

It remains to analyze the running time. The block $B_1$ can be found by a
root-to-leaf traversal in the B-tree, taking
$O(F\cdot \BTHeight{i})$ time. The enumerations of the blocks $B_1$ and
$B_2$ are supported by Lemma~\ref{lem-block-PA-LCS}, and together take
$O(\BlockTime{i})$ time. The upward traversal
inspects at most $O(F)$ siblings per level, and the downward traversal
inspects at most $O(F)$ children per level. Since the B-tree has height
$\BTHeight{i}$, these traversals take
$O(F\cdot \BTHeight{i})$ time in total. Hence the total time for
$NSV_i(j,\mathit{val})$ is
$O(F\cdot \BTHeight{i}+\BlockTime{i})\subseteq O(\BlockTime{i})$. The same bound holds for
$PSV_i(j,\mathit{val})$ by symmetry.
\end{proof}

\subsection{Finding the $PA$-interval of any substring of $T[1..i]$}\phantomsection\label{def-pa-interval-computation}
In Lemma~\ref{lem-pa-interval}, we show how to compute the $PA$-interval of a substring of $T[1..i]$ and how to find its leftmost occurrence.

\begin{lemma}\label{lem-pa-interval}
In each round $i$, one can maintain a data structure using
$O(w(r_i+i/L))$ bits of space such that, given an interval
$[j_1..j_2]\subseteq [1..i]$ and the value $PA^{-1}[h]$ for any ending position
$h$ with $T[h-(j_2-j_1)..h]=T[j_1..j_2]$, one can compute the
$PA$-interval of $T[j_1..j_2]$ and find its leftmost occurrence in
$T[1..i]$ in $O(\BlockTime{i})$ time.
\end{lemma}

\begin{proof}
Let $X=T[j_1..j_2]$.
By the definition of $h$, $X$ is a suffix of the prefix $T[1..h]$.
Let $[e..f]$ denote the $PA$-interval of $X$, and let
$q=PA^{-1}[h]$.
Then, it follows that $q\in [e..f]$.

Let $m=|X|$. 
Observe that the $PA$-interval of $X$ is the maximal interval
containing $q$ such that all adjacent $LCS$-values inside the interval are at least $m$.
Hence, if $PSV_i(q,m)=\bot$, then $e=1$; otherwise, $e=PSV_i(q,m)$. 
Symmetrically, if
$q=i$, then $f=i$. Otherwise, if $NSV_i(q+1,m)=\bot$, then $f=i$; and
otherwise, $f=NSV_i(q+1,m)-1$. Therefore, by Lemma~\ref{lem-nsv}, the
interval $[e..f]$ can be computed in $O(\BlockTime{i})$ time.

All occurrences of $X$ in $T[1..i]$ correspond exactly to the prefixes of $T[1..i]$
whose right endpoints are stored in $PA[e..f]$. Thus the leftmost occurrence of
$X$ is obtained by taking the minimum $PA$-value in this range. By
Lemma~\ref{lem-RMQ}, we can compute
$j_{\min}=\min\{PA[k]:k\in[e..f]\}$ in
$O(\BlockTime{i})$ time. Since the occurrence of $X$ ending
at position $j_{\min}$ starts at position $j_{\min}-m+1$, this gives the
leftmost occurrence of $X$ in $T[1..i]$.

The total time is $O(\BlockTime{i})$, and the claimed space
bound follows from the maintained data structures of
Lemmas~\ref{lem-block-PA-LCS}, \ref{lem-RMQ}, and~\ref{lem-nsv}.
\end{proof}

\section{Applications}
\label{sect-app}

In this section, we present two applications of the data structures developed in the previous sections: the online computation of the value $LRS[i]$ and the online maintenance of smallest suffixient sets. Both applications rely on the following components:
\begin{itemize}
    \item Lemma~\ref{lem-rlBWT-online}, which gives the online representation of the run-length BWT;
    \item Lemmas~\ref{lem-block-PA-LCS} and~\ref{lem-update-fusion}, which support the enumeration of $PA$- and $LCS$-entries within an arbitrary block and the online update of the sampled $PA$- and $LCS$-information, respectively; and
    \item Lemmas~\ref{lem-RMQ} and~\ref{lem-nsv}, which provide range minimum queries over the $LCS$ and $PA$ arrays, as well as $NSV$ and $PSV$ queries over the $LCS$ array.
\end{itemize}

In particular, Lemma~\ref{lem-RMQ} supplies the range minimum queries required by Lemma~\ref{lem-update-fusion}. Therefore, after substituting the bound of Lemma~\ref{lem-RMQ} into Lemma~\ref{lem-update-fusion}, the maintained data structures can be updated from round $i-1$ to round $i$ in $O(\log r_i+\BlockTime{i})$ time. The total space usage of the maintained data structures is $O(w(r_i+i/L))$ bits in round $i$.

\subsection{Maintaining the LRS online}
\label{def-online-lpfa}

We explain how to compute $LRS[i]$ in each round $i$, after updating the maintained data structures for $T[1..i]$. As shown above, these updates take $O(\log r_i+\BlockTime{i})$ time.
Following the idea of Okanohara and Sadakane~\cite{OkanoharaS08}, we compute $LRS[i]$ by comparing the longest common suffixes of $T[1..i]$ with its predecessor and successor in co-lexicographic order. Our contribution is to support the required accesses using the compressed data structures developed in the preceding sections, while guaranteeing worst-case bounds.

If $i\le 2$, we
set $LRS[i]\leftarrow 0$. Otherwise, let $\ell$ denote the position of the
newly inserted entry in $LCS(T[1..i])$, or equivalently, the co-lexicographic
rank of the new prefix $T[1..i]$ among all prefixes of $T[1..i]$.
If $\ell=i$, the new prefix has no successor in the co-lexicographic
order, and we set $LRS[i]\leftarrow LCS(T[1..i])[\ell]$. Otherwise, we set
$LRS[i]\leftarrow \max\{LCS(T[1..i])[\ell],LCS(T[1..i])[\ell+1]\}$. By
Lemma~\ref{lem-block-PA-LCS}, the entries $LCS(T[1..i])[\ell]$ and, if
$\ell<i$, $LCS(T[1..i])[\ell+1]$ can be accessed in
$O(\BlockTime{i})$ time.

The correctness follows from the fact that, in a co-lexicographically sorted set of strings, the longest common suffix between a fixed string and
any other string in the set is attained by one of its adjacent strings in
the sorted order. 
In our setting, the fixed string is the newly inserted prefix $T[1..i]$. Every other stored prefix is of the form $T[1..j]$ for
some $j<i$, and therefore any common suffix of $T[1..i]$ and $T[1..j]$ corresponds to a suffix of $T[1..i]$ that also occurs ending at an earlier
position $j$. Conversely, every previous occurrence of a suffix ending at $i$ is represented by such a prefix $T[1..j]$ with $j<i$. Thus $LRS[i]$ is exactly the maximum longest common suffix between $T[1..i]$ and any other stored prefix. Since $T[1..i]$ has rank $\ell$, its only possible adjacent
prefixes are the prefixes at ranks $\ell-1$ and $\ell+1$, whenever they exist. The corresponding longest common suffix values are stored in
$LCS(T[1..i])[\ell]$ and $LCS(T[1..i])[\ell+1]$, respectively. 
This establishes Theorem~\ref{theorem-LPFA}.


\begin{theorem}\label{theorem-LPFA}
In each round $i$, one can maintain a data structure using
$O(w(r_i+i/L))$ bits of space that computes $LRS[i]$ in
$O(\log r_i+\BlockTime{i})$ time.
\end{theorem}

Note that our online algorithm reports only the current value $LRS[i]$ in round $i$; it does not store the previously computed $LRS$ values. In the next section, we present an application for which access to the current value $LRS[i]$ alone is sufficient.

By choosing appropriate values of the block size $L$ and the fan-out $F$,
we obtain the following trade-offs as shown in Corollary~\ref{cor-LPFA}. 

\begin{restatable}{corollary}{lrsTradeoffs}\label{cor-LPFA}
In each round \(i\), computing \(LRS[i]\) admits either trade-off below:
\begin{itemize}
\item[(a)] $O(r_i\log n+i)$ bits of space and $O(\log^2 n/\log\log n)$ time; and
\item[(b)] $O(r_i\log n+i\log\log n)$ bits of space and $O((\log n/\log\log n)^2)$ time.
\end{itemize}
\end{restatable}

\begin{proof}
Recall that $w=\Theta(\log n)$ and $r_i\le n$. Hence
$\log_w r_i=O(\log n/\log\log n)$ and $\log r_i=O(\log n)$.
By Theorem~\ref{theorem-LPFA}, the running time is $O(\BlockTime{i}+\log r_i)=O(\log r_i+FH_i+L\cdot(1+\log_w r_i))$, where $H_i=1+\log_F(1+i/L)$.

For the first trade-off, set $F=4$ and $L=\Theta(\log n)$. The space bound
of Theorem~\ref{theorem-LPFA} becomes
$O(w(r_i+i/L))=O(r_i\log n+i)$ bits. The running time becomes
$O(\log r_i+1+\log(1+i/L)+L\cdot(1+\log_w r_i))$. Since $i\le n$ and
$L=\Theta(\log n)$, we have $\log r_i+1+\log(1+i/L)=O(\log n)$. Moreover,
$L\cdot(1+\log_w r_i)=O(\log n\cdot(1+\log n/\log\log n))
=O(\log^2 n/\log\log n)$. Thus the total time is
$O(\log^2 n/\log\log n)$.

For the second trade-off, set $L=\Theta(\log n/\log\log n)$ and
$F=\log^\epsilon n$, where $0<\epsilon<1$ is a fixed constant. The space
bound becomes $O(w(r_i+i/L))=O(r_i\log n+i\log\log n)$ bits. The term
$L\cdot(1+\log_w r_i)$ is
$O((\log n/\log\log n)\cdot(1+\log n/\log\log n))
=O((\log n/\log\log n)^2)$. Since
$\log_F(1+i/L)=O(\log n/\log\log n)$, we also have
$F(1+\log_F(1+i/L))=O(\log^\epsilon n\cdot \log n/\log\log n)
=O((\log n/\log\log n)^2)$, where the last bound uses $\epsilon<1$. Finally,
the additive term $\log r_i=O(\log n)$ is dominated by
$O((\log n/\log\log n)^2)$. Therefore, the total time is
$O((\log n/\log\log n)^2)$.
\end{proof}




By Proposition~\ref{pro-bound-on-r} and Corollary~\ref{cor-LPFA}, the two trade-offs use $O(r\log n+n)$ bits and $O(r\log n+n\log \log n)$ bits of working space, respectively, throughout the execution of the online algorithm, where $r$ is the number of runs in $BWT(\overleftarrow{T[1..n]})$.
Before concluding this subsection, we establish a lower bound on the working space required for online LRS computation.

\begin{theorem}\label{theorem-lower-bound-LRS} Let \(n\ge 4\). Any deterministic online algorithm that correctly computes \(LRS[j]\) in round \(j\), for every \(j\in[1..n]\) and every text \(T[1..n]\) over \(\{\$, \#, 0, 1\}\) where \(\$\) occurs only at position \(1\), requires \(\Omega(n)\) bits of peak working space in the worst case. \end{theorem}

\begin{proof}
To this end, we prove that for every \(i\ge 2\) and \(n\ge 2i\), there exists an input of length \(n\) on which the algorithm requires \(\Omega(i)\) bits of working space after round \(i\).

Fix integers \(i\ge 2\) and \(n\ge 2i\), and let \(m=i-1\).
Let $X, Y\in \{0, 1\}^m$ denote any two distinct binary strings. 
Consider two instances that, respectively, read $\$\cdot X$ and $\$\cdot Y$ as the input text in the first $m+1$ rounds.
We prove that the memory states of both instances must be different upon completion of round $m+1$.

Assume, for the sake of contradiction, that the memory states of both instances are the same, after completing round $i$.
Now continue both executions with the same sequence $\#\cdot X$, followed by the same arbitrary string \(Z\in\{0,1\}^{n-2i}\). Since the algorithm is deterministic and the two executions have the same memory state before reading $\#\cdot X$, they must produce the same output in every subsequent round, in particular in round $2m+2=2i$.

However, in the first instance, the input prefix read by round \(2m+2\) is \(T_1[1..2m+2]=\$ X\#X\). The value of $LRS[2m+2]$ is at least $m$, since the suffix $X$ of length $m$ appears in the prefix $T_1[1..m+1]$.
In the second instance, the corresponding input prefix is \(T_2[1..2m+2]=\$Y\#X\), and the value of $LRS[2m+2]$ must be smaller than $m$, since every length-$m$ substring of $T_2[1..2m+1]$ is either $Y$ or a substring containing $\$$ or $\#$, each of which is different from $X$.
Both LRS values must be different in round $2m+2$, a contradiction.

Since \(X\) and \(Y\) are arbitrary distinct binary strings, the algorithm has at least \(2^m\) possible memory states after round \(m+1=i\), and hence requires at least \(m=i-1=\Omega(i)\) bits in the worst case. Finally, choosing \(i=\lfloor n/2\rfloor\) gives \(\Omega(n)\) bits of peak working space.
\end{proof}

\subsection{Online construction of smallest suffixient sets}
\label{sect-online-consctruct-sss}

\phantomsection\label{def-sss-records}



%
For every round \(i\), let \(T_i=T[1..i]\) and \(R_i=\overleftarrow{T_i}\). Our goal is to maintain the length-annotated rightmost SSS, denoted by \(\widehat{{RSSS}}(R_i)\). For algorithmic convenience, we instead maintain its \(T_i\)-coordinate representation $	\mathcal{A}_i=\{\langle i-q+1,\Delta\rangle : \langle q,\Delta\rangle \in\widehat{{RSSS}}(R_i)\}$. Thus, a record \(\langle p,\Delta\rangle\in\mathcal{A}_i\) represents the occurrence \(T_i[p..p+\Delta-1]\), whose reversal is the selected rightmost occurrence in \(R_i\).
Our result is summarized as Theorem~\ref{theorem-SSS}. 

\begin{theorem}\label{theorem-SSS}
In each round $i$, $\Aset{i-1}$ can be updated to $\Aset{i}$ in
$O(\log r_i+\BlockTime{i})$ time using $O(w(r_i+i/L))$ bits of space.
Moreover, the positions in the smallest suffixient set for
$\overleftarrow{T[1..i]}$ can be enumerated in
$O(|\Aset{i}|(1+\log_w r_i))$ time.
\end{theorem}

Recall that $\BlockTime{i}=F\cdot\BTHeight{i}+L\cdot(1+\log_w r_i)$ and $\BTHeight{i}=1+\log_F(1+i/L)$. By choosing appropriate values of the block size $L$ and the fan-out $F$, as in Corollary~\ref{cor-LPFA}, we obtain, 

\begin{corollary}\label{cor-SSS}
In each round \(i\), updating \(\Aset{i-1}\) to \(\Aset{i}\) admits either trade-off below:
\begin{itemize}
\item[(a)] $O(r_i\log n+i)$ bits of space and $O(\log^2 n/\log\log n)$ update time; and
\item[(b)] $O(r_i\log n+i\log\log n)$ bits of space and $O((\log n/\log\log n)^2)$ update time.
\end{itemize}
\end{corollary}

By Proposition~\ref{pro-bound-on-r} and Corollary~\ref{cor-SSS}, the two trade-offs use $O(r\log n+n)$ bits and $O(r\log n+n\log\log n)$ bits of working space, respectively, throughout the execution of the online algorithm, where $r$ is the number of runs in $BWT(\overleftarrow{T[1..n]})$.


\subsubsection{The high-level idea of the algorithm}
\label{def-sss-high-level}
Our online construction of smallest suffixient sets is inspired by the algorithm
of~\cite[Section~4.1]{koppl2026smallestsuffixientsetmaintenance}. 
In each round $i\ge 1$, the records $\langle p,\Delta \rangle \in \Aset{i}$ are stored in
a dynamic fusion tree $\mathcal{B}$, keyed lexicographically by the pair
$\langle p,\Delta \rangle$. We encode each pair $\langle p,\Delta \rangle$ as the integer
$(p-1)(n+1)+\Delta$, which preserves the lexicographic order of the pairs
and uses $O(\log n)$ bits. Thus, insertions, deletions, and membership
queries on $\mathcal{B}$ take $O(1+\log_w |\Aset{i}|)$ time. As
$|\Aset{i}|=O(r_i)$, each such operation takes $O(1+\log_w r_i)$ time.

In the first round \((i=1)\), 
we initialize $\mathcal{B}$ with the record $\langle 1, 1 \rangle$, regarding it as $\Aset{1}$ for convenience, although a smallest suffixient set is defined only for a text that contains at least two distinct characters.
Now consider a round \(i\ge2\).  Let
\(c=T[i]\), let \(T_{i-1}=T[1..i-1]\), and let
\(R_{i-1}=\overleftarrow{T_{i-1}}\).  We start from \(\mathcal{B}\) storing \(\Aset{i-1}\), and update it so that it stores
\(\Aset{i}\).

If \(c\) does not occur in \(T_{i-1}\), then the new character creates the
new SRE \((\varepsilon,c)\) for \(R_i=c\cdot R_{i-1}\), where $\varepsilon$ denotes an empty string.  In the coordinate
system of \(T_i\), this SRE is represented by the length-one occurrence
\(T_i[i]\).  Hence, we insert the record \(\langle i,1 \rangle\) into \(\mathcal{B}\),
regard the resulting tree as \(\Aset{i}\), and proceed to the next round.

Otherwise, $c$ occurs in \(T_{i-1}\).
Let \(P\) be the longest suffix of \(T_{i-1}\), possibly empty, such that
\(P\cdot c\) occurs in \(T_{i-1}\), and let \(Q=\overleftarrow{P}\).
Equivalently, \(Q\) is the longest prefix of \(R_{i-1}\) such that
\(c\cdot Q\) occurs in \(R_{i-1}\).  Let
\(y=T[i-1-|P|]\), so that \(y\cdot P\) is a suffix of \(T_{i-1}\).

\begin{figure}[t]
    \centering
    \includegraphics[width=1\textwidth]{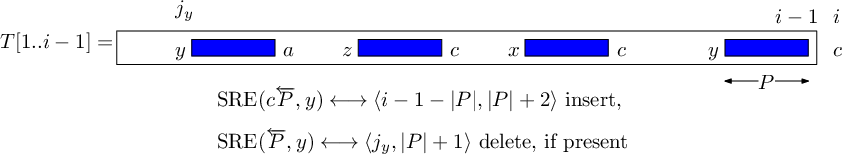}
    \caption{Illustration of the right-maximal case. Each blue block represents an occurrence of $P$. The occurrences $z\cdot P\cdot c$ and $x\cdot P\cdot c$, where $x\ne z$, show that $c\cdot\overleftarrow{P}$ has at least two distinct right extensions in $R_{i-1}$. Appending $c=T[i]$ creates the suffix occurrence $y\cdot P\cdot c$. Accordingly, the record representing $(c\cdot\overleftarrow{P},y)$ is inserted, while the record representing $(\overleftarrow{P},y)$ is deleted if present.}
    \label{fig:right-maximal}
\end{figure}

\begin{itemize}
    \item \textbf{Step 1: New-SRE insertion.}
\label{step-SRE-insert}
By~\cite[Lemma~8(a)]{koppl2026smallestsuffixientsetmaintenance},
$(c\cdot Q,y)=(\overleftarrow{P\cdot c},y)$ is a new SRE of $R_i=c\cdot R_{i-1}$.
In the coordinate system of \(T_i\),
this SRE is represented by the occurrence \(y\cdot P\cdot c\), which starts
at position \(i-1-|P|\) and has length \(|P|+2\).  We therefore insert the record $\langle i-1-|P|,\ |P|+2\rangle$ into \(\mathcal{B}\).
    \item \textbf{Step 2: Right-extension classification.}
\label{step-classification}
We check whether $c\cdot Q=\overleftarrow{P\cdot c}$ is right-maximal in $R_{i-1}$ and perform one of the following two updates accordingly:
\begin{itemize}
    \item \textbf{Case 1: Right-maximal case.} \label{case:right-maximal}
If \(c\cdot Q=\overleftarrow{P\cdot c}\) has at least two distinct right
extensions in \(R_{i-1}\), corresponding to \cite[Lemma 8(a)]{koppl2026smallestsuffixientsetmaintenance}, then we search for the leftmost occurrence
\(j_y\) of \(y\cdot P\) in \(T_{i-1}\), and remove the record $\langle j_y,\ |P|+1 \rangle$ from \(\mathcal{B}\), if it is present. An illustration of this case is given in Figure~\ref{fig:right-maximal}. 
    \item \textbf{Case 2: Unique-extension case.} \label{case:unique-ext}
Otherwise, all occurrences of \(c\cdot Q=\overleftarrow{P\cdot c}\) in
\(R_{i-1}\) are followed by the same character \(z\ne y\), corresponding to \cite[Lemma 8(b)]{koppl2026smallestsuffixientsetmaintenance}.  
In this case, we first
search for the leftmost occurrence \(j_{zc}\) of \(z\cdot P\cdot c\) in
\(T_{i-1}\), and insert the record $\langle j_{zc},\ |P|+2\rangle$
into \(\mathcal{B}\).
We then search for the leftmost occurrences
\(j_y\) and \(j_z\) of \(y\cdot P\) and \(z\cdot P\) in \(T_{i-1}\),
respectively and remove each of the records $\langle j_y,\ |P|+1\rangle$ and $\langle j_z,\ |P|+1\rangle$ from \(\mathcal{B}\), if it is present.
See Figure~\ref{fig:unique-case} for an illustration. 
\end{itemize}
\end{itemize}

After these updates, $\mathcal{B}$ represents
\(\Aset{i}\), and the algorithm proceeds to the next round.

\subsubsection{The correctness of the online algorithm for smallest suffixient sets}
In Lemma~\ref{lem-correctness-sss}, we show that our right-maximal (\hyperref[case:right-maximal]{Case~1}) and unique-extension (\hyperref[case:unique-ext]{Case~2}) cases perform exactly the SRE updates corresponding to the hard-link and soft-link cases of~\cite[Lemma~8]{koppl2026smallestsuffixientsetmaintenance}, respectively. We also prove that every maintained record stores a leftmost representative occurrence.



\begin{lemma}\label{lem-correctness-sss}
For every \(i> 1\), the set $  \{\, i-p+1 : \langle p,\Delta \rangle \in \Aset{i}\,\}$
is a smallest suffixient set for \(R_i=\overleftarrow{T[1..i]}\).
Moreover, each record \(\langle p,\Delta \rangle \in \Aset{i}\) represents the
leftmost occurrence \(T[p..p+\Delta-1]\) in \(T[1..i]\). Equivalently, it corresponds
to the SRE $\left(\overleftarrow{T[p+1..p+\Delta-1]},\,T[p]\right)$
of \(\overleftarrow{T[1..i]}\).
\end{lemma}

\begin{figure}[t]
    \centering
    \includegraphics[width=1\textwidth]{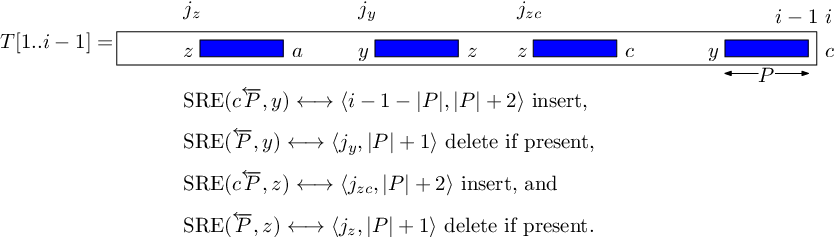}
\caption{Illustration of the unique-extension case. Each blue block represents an occurrence of $P$. In this  instance, $P\cdot c$ occurs only once in $T[1..i-1]$, and this occurrence is preceded by the character $z\ne y$. The positions $j_y$, $j_z$, and $j_{zc}$ denote the leftmost occurrences of $y\cdot P$, $z\cdot P$, and $z\cdot P\cdot c$, respectively. Appending $c=T[i]$ creates the suffix occurrence $y\cdot P\cdot c$. Accordingly, the records representing $(c\cdot\overleftarrow{P},y)$ and $(c\cdot\overleftarrow{P},z)$ are inserted, while those representing $(\overleftarrow{P},y)$ and $(\overleftarrow{P},z)$ are deleted if present.}
    \label{fig:unique-case}
\end{figure}

\begin{proof}
For \(i=2\), since \(T[1]=\$\) and \(\$\) does not occur elsewhere, we have
\(T[2]\neq \$\).  The algorithm inserts the record \(\langle 2,1 \rangle\), and the tree $\mathcal{B}$
contains exactly the two records \(\langle 1,1\rangle\) and \(\langle 2,1 \rangle\).  These records
represent the two SREs \((\varepsilon,\$)\) and
\((\varepsilon,T[2])\) of \(R_2=\overleftarrow{T[1..2]}\), respectively, where $\varepsilon$ denotes the empty string.
Both records clearly store the leftmost occurrences of the represented extensions
in \(T[1..2]\).  Thus the statement holds for \(i=2\).

    Assume now that the statement holds after round \(i-1\), where \(i>2\).  Let
\(c=T[i]\), \(T_{i-1}=T[1..i-1]\), and \(R_{i-1}=\overleftarrow{T_{i-1}}\).
We prove that, after processing \(c\), the resulting tree $\mathcal{B}$ represents
\(\Aset{i}\).

Suppose that \(c\) does not occur in \(T_{i-1}\).  Then, in
\(R_i=cR_{i-1}\), the only new SRE is
\((\varepsilon,c)\).  Indeed, any nonempty string beginning with \(c\) occurs
only once in \(R_i\), and therefore cannot be right-maximal; consequently, it
cannot create any new SRE other than
\((\varepsilon,c)\), nor can it destroy an old one. 
The algorithm inserts $\langle i,1\rangle$, which represents the unique, and hence leftmost, occurrence $T[i]=c$.  
All old records remain correct by the induction
hypothesis, and hence the statement holds in this case.

	It remains to consider the case that $c$ occurs in $T_{i-1}$. Let $P$ be the longest suffix of $T_{i-1}$ such that $P\cdot c$ occurs in $T_{i-1}$, and let $Q=\overleftarrow{P}$. We apply~\cite[Lemma~8]{koppl2026smallestsuffixientsetmaintenance} to the right-to-left update from $R_{i-1}$ to $R_i=c\cdot R_{i-1}$. Under our terminology, the hard-link case of~\cite[Lemma~8(a)]{koppl2026smallestsuffixientsetmaintenance} is the right-maximal case (\hyperref[case:right-maximal]{Case~1}), in which $c\cdot Q=\overleftarrow{P\cdot c}$ has at least two distinct right extensions in $R_{i-1}$. The soft-link case of~\cite[Lemma~8(b)]{koppl2026smallestsuffixientsetmaintenance} is the unique-extension case (\hyperref[case:unique-ext]{Case~2}), in which all occurrences of $c\cdot Q$ in $R_{i-1}$ are followed by the same character $z$.
    
    Let $y=T[i-1-|P|]$, so that the prefix occurrence of $c\cdot Q$ in $R_i$ is followed by $y$.
    By~\cite[Lemma 8]{koppl2026smallestsuffixientsetmaintenance}, the only changes in the set of SREs are the following: In the right-maximal case (\hyperref[case:right-maximal]{Case~1}), the new SRE $(cQ,y)$ is added and $(Q,y)$ is deleted if it was present. In the unique-extension case (\hyperref[case:unique-ext]{Case~2}), the two new SREs $(cQ,y)$ and $(cQ,z)$ are added, where $z\ne y$, and $(Q,y)$ and $(Q,z)$ are deleted if they were present.
    We now verify that the algorithm performs these corresponding updates.

		In the coordinate system of $T_i$, the SRE $(cQ,y)$ is represented by the occurrence $yPc$, whose start position is $i-1-|P|$ and whose length is $|P|+2$. This occurrence is unique, and hence leftmost: if $yPc$ occurred earlier in $T_{i-1}$, then $yP$ would be a longer suffix of $T_{i-1}$ whose concatenation with $c$ occurs in $T_{i-1}$, contradicting the maximality of $P$. Hence the inserted record $\langle i-1-|P|, |P|+2\rangle$ is correct.
		
		In the unique-extension case (\hyperref[case:unique-ext]{Case~2}), the additional new SRE $(cQ,z)$ is represented in $T_i$ by an occurrence of $zPc$. The algorithm explicitly searches for the leftmost such occurrence $j_{zc}$ in $T_{i-1}$ and inserts the record $\langle j_{zc}, |P|+2\rangle$, so this insertion is also correct. 

        Consider now an old SRE $(Q,a)$, where $a\in\{y,z\}$, that ceases to be supermaximal during the update. In the coordinate system of $T_{i-1}$, it is represented by the string $a\cdot P$. By the induction hypothesis, $\mathcal{B}$ contains the record
$\langle j_a,\ |P|+1\rangle$,
where $j_a$ is the leftmost occurrence of $a\cdot P$ in $T_{i-1}$.
        Therefore, the algorithm
        finds this record and deletes it correctly.
		
		Following \cite[Lemma~8]{koppl2026smallestsuffixientsetmaintenance}, no other SRE is inserted or deleted. All records not explicitly deleted remain valid by the induction hypothesis, and their stored occurrences remain leftmost because appending $c$ to the end of $T_{i-1}$ cannot create an earlier occurrence of any old represented string. Hence, after the update, $\Aset{i}$ contains exactly one leftmost representative record for every SRE of $R_i$, and therefore $\{i-p+1 : \langle p,\Delta\rangle\in \Aset{i}\}$ is a smallest suffixient set for $R_i$.
        This completes the induction.
\end{proof}

\subsubsection{The detailed description of the algorithm and its complexity} 
We explain how the queries required to perform these updates are supported by our data structures.

Let $c=T[i]$.
Henceforth, assume that our incremental BWT-based index has already been updated to round $i$ and therefore represents the current prefix $T[1..i]$.

As shown in the previous section, the first step in round $i$ is to decide whether $c$ appears in $T[1..i-1]$.
Observe that $c$ appears in $T[1..i-1]$ if and only if the number of occurrences of $c$ in $T[1..i]$ is at least two.
Hence, Lemma~\ref{lem-rlBWT-online} is sufficient to count the number of occurrences of $c$ in $T[1..i]$, thereby answering this query, taking $O(\log r_i)$ time.

If $c$ does not appear in $T[1..i-1]$, then we simply insert $\langle i, 1\rangle$ into $\mathcal{B}$ that stores records in $\Aset{i-1}$, taking $O(1+\log_w |\Aset{i-1}|)\subseteq O(1+\log_w r_i)$ time, since $|\Aset{i-1}|\le 2r_{i-1} \subseteq O(r_i)$.  Then, proceed to the next round.

Henceforth, assume $c$ appears in $T[1..i-1]$.
To perform Step~1 of the algorithm, the new-SRE insertion, we first determine the length of the longest suffix $P$ of $T[1..i-1]$ such that $P\cdot c$ occurs in $T[1..i-1]$.
Observe that $|P|=LRS[i]-1$.
Therefore, by Theorem~\ref{theorem-LPFA}, $|P|$ can be computed in $O(\log r_i+\BlockTime{i})$ time.
Once $|P|$ is known, inserting the record $\langle i-1-|P|, |P|+2 \rangle$ into $\mathcal{B}$ takes $O(1+\log_w |\Aset{i-1}|)\subseteq O(1+\log_w r_i)$ time.

To perform Step~2 of the algorithm, the right-extension classification, we determine whether $\overleftarrow{P\cdot c}$ has at least two distinct right extensions in $\overleftarrow{T[1..i-1]}$. To this end, we first compute the $PA$-interval $[e..f]$ of $P\cdot c$ as follows.

Let $\ell$ denote the position of the newly added $LCS$ entry in this round, which is also the colexicographical rank of $T[1..i]$ among all the prefixes of $T[1..i]$ (see details in Section~\ref{sect-update-B-tree-Fusion-Tree}).
Observe that $P\cdot c$ is a suffix of $T[1..i]$ and $\ell=PA^{-1}[i]$.
By applying Lemma~\ref{lem-pa-interval}, setting the query parameters $[j_1, j_2]\leftarrow [i-|P|..i]$, the interval $[e..f]$ can be found in $O(\BlockTime{i})$ time.

Note that $\ell\in [e..f]$.
Let $y= T[i-|P|-1]$.
If $e<\ell <f$, then there exist two substrings $x\cdot P \cdot c$ and $z\cdot P \cdot c$ with $x<y<z$ in $T[1..i-1]$; this implies that  \(\overleftarrow{P\cdot c}\) has at least two distinct right extensions in \(\overleftarrow{T[1..i-1]}\).

Otherwise, either $\ell=f$ or $\ell=e$.
Observe that if \(\overleftarrow{P\cdot c}\) has at least two distinct right extensions in \(\overleftarrow{T[1..i-1]}\), then their longest common prefix is exactly \(\overleftarrow{P\cdot c}\) of length $|P|+1$.
Hence, if $\ell=f$, then \(\overleftarrow{P\cdot c}\) has at least two distinct right extensions in \(\overleftarrow{T[1..i-1]}\) if and only if $f>e+1$ and $\min \{LCS_i[e+1],\dots, LCS_i[f-1]\}=|P|+1$, where $LCS_i$ denotes $LCS(T[1..i])$.
If $\ell=e$, the symmetric test is $f>e+1$ and $\min\{LCS_i[e+2],\ldots,LCS_i[f]\}=|P|+1$.

By Lemma~\ref{lem-RMQ}, a range minimum query over the array $LCS_i$ can be supported in $O(\BlockTime{i})$ time.
Therefore, we can decide whether \(\overleftarrow{P\cdot c}\) has at least two distinct right extensions in \(\overleftarrow{T[1..i-1]}\), spending $O(\BlockTime{i})$ time.

If \(\overleftarrow{P\cdot c}\) has at least two distinct right extensions in \(\overleftarrow{T[1..i-1]}\), then we are in the right-maximal case (\hyperref[case:right-maximal]{Case~1}). 
As prescribed before, we need to find the leftmost occurrence $j_y$ of $y\cdot P$ in $T[1..i-1]$, which is also the leftmost occurrence in $T[1..i]$, since $c$ is appended to the right-end of $T[1..i-1]$.

Recall that $y\cdot P=T[i-1-|P|..i-1]$ and that $\ell=PA^{-1}[i]$. 
The position \(PA^{-1}[i-1]\) can be computed by the \(FL\)-mapping in the FM-index~\cite{FerraginaM00}: if \(x\) is the number of occurrences in \(BWT(\overleftarrow{T[1..i]})\) of symbols that are smaller than \(c\), then \(PA^{-1}[i-1]=\select_c(BWT(\overleftarrow{T[1..i]}), \ell-x)\).
Hence,  \(PA^{-1}[i-1]\) can be computed in $O(\log r_i)$ time by Lemma~\ref{lem-rlBWT-online}.
Once \(PA^{-1}[i-1]\) is known, we can apply Lemma~\ref{lem-pa-interval} to find the leftmost occurrence $j_y$ of $y\cdot P$ in $T[1..i-1]$, taking $O(\BlockTime{i})$ time.
After finding $j_y$, we check whether the record $\langle j_y, |P|+1\rangle$ appears in $\mathcal{B}$; if it is present, we delete it from $\mathcal{B}$, as prescribed by the right-maximal-case (\hyperref[case:right-maximal]{Case~1}).
This case takes overall $O(\log r_i+\BlockTime{i})$ time in round $i$.

Otherwise, all occurrences of $\overleftarrow{P\cdot c}$ in $\overleftarrow{T[1..i-1]}$ are followed by the same character $z\ne y$, and we are therefore in the unique-extension case (\hyperref[case:unique-ext]{Case~2}). 
In this case, we also need to compute the record $\langle j_y, |P|+1\rangle$ as the right-maximal case (\hyperref[case:right-maximal]{Case~1}).
In addition, as prescribed by the unique-extension case (\hyperref[case:unique-ext]{Case~2}) of the algorithm, we need to find the leftmost occurrences of $z \cdot P\cdot c$ and $z\cdot P$ in $T[1..i-1]$, which are also the leftmost occurrences in $T[1..i]$.

For the leftmost occurrence $j_{zc}$ of $z \cdot P\cdot c$ in $T[1..i]$, recall that $[e..f]$ is the PA-interval of $P\cdot c$ and that $y\cdot P\cdot c$ has a unique occurrence in $T[1..i]$.
Hence, $\ell$, the position of the newly added LCS entry, has to be either $e$ or $f$.
Observe that if $\ell=f$, then the PA-interval of $z\cdot P \cdot c$ is exactly $[e..f-1]$, since all occurrences of $\overleftarrow{P\cdot c}$ in $\overleftarrow{T[1..i-1]}$ are followed by the same character $z\ne y$; otherwise, this PA-interval is $[e+1, f]$.
With the PA-interval, we can compute the leftmost occurrence $j_{zc}$ of $z\cdot P \cdot c$ in $T[1..i]$, as well as the record $\langle j_{zc}, |P|+2\rangle$, following the same method shown earlier.
Then, we insert the record into $\mathcal{B}$.
This takes $O(\log r_i+\BlockTime{i})$ time in round $i$.

For the leftmost occurrence $j_z$ of $z \cdot P$ in $T[1..i]$, we need to find the PA-interval $z\cdot P$ in $T[1..i]$. To this end, we first find one position $h$ in this PA-interval such that $z\cdot P$ is a suffix of $T[1..PA[h]]$. Assume first that $\ell=f$. Since $[e..f]$ is the PA-interval of $P\cdot c$, the row $f$ corresponds to the unique occurrence $y\cdot P\cdot c$ as a suffix of $T[1..i]$. Moreover, all occurrences of $\overleftarrow{P\cdot c}$ in $\overleftarrow{T[1..i-1]}$ are followed by the same character $z\ne y$. Hence, $f-1$ is a position in the PA-interval of $z\cdot P\cdot c$ (Symmetrically, if $\ell=e$, then $e+1$ is a position in the PA-interval of $z\cdot P\cdot c$.)  
Let $p=PA_i[f-1]$, where $PA_i$ denotes $PA(T[1..i])$. Then $z\cdot P\cdot c$ is a suffix of $T[1..p]$, and therefore \(z\cdot P\) is a suffix of \(T[1..p-1]\). Thus \(h=PA_i^{-1}[p-1]=PA_i^{-1}[PA_i[f-1]-1]\) is a position in the PA-interval of \(z\cdot P\).

The position \(h\) can be computed again by the \(FL\)-mapping: if \(x\) is the number of occurrences in \(BWT(\overleftarrow{T[1..i]})\) of symbols that are smaller than \(c\), then \(h=\select_c(BWT(\overleftarrow{T[1..i]}), f-1-x)\). By Lemma~\ref{lem-rlBWT-online}, this takes \(O(\log r_i)\) time. 
The case \(\ell=e\) is symmetric: we use the row \(e+1\) instead of \(f-1\) and set $p= PA_i[e+1]$.
%
With \(h\), one position in the PA-interval of \(z\cdot P\), we can compute the whole PA-interval and the leftmost occurrence $j_z$ of $z\cdot P$ in $T[1..i]$ by Lemma~\ref{lem-pa-interval}, setting the query interval $[j_1..j_2]\leftarrow[p-|P|-1..p-1]$. 
Finally, we remove the record \(\langle j_z, |P|+1\rangle\) from \(\mathcal{B}\) if it is present.
This procedure also takes $O(\log r_i+\BlockTime{i})$ time.

As a result, we achieve Theorem~\ref{theorem-SSS} shown in the beginning of Section~\ref{sect-online-consctruct-sss}.

\subsubsection{Lower bounds for online suffixient-set maintenance} \label{def-sss-reduction}

We derive two lower bounds through reductions from
online LRS computation. First, we show that maintaining any
length-annotated SSS of \(R_i=\overleftarrow{T_i}\) is sufficient to
compute \(\mathrm{LRS}[i]\). 
Let \(A_i\) be any smallest suffixient set of \(R_i\), not necessarily the rightmost SSS maintained by our algorithm.
We denote by \(\widehat{A}_i\) its length-annotated SSS, expressed in the coordinate
system of \(T_i\). Formally,
\(\widehat{A}_i=\{\langle i-k+1,\Delta\rangle : k\in A_i\}\), where \(\Delta\) is the length of the SRE represented by the occurrence \(R_i[k-\Delta+1..k]\). This reduction is established in Lemma~\ref{lem-reduction}. Then, we show, in Theorem~\ref{theorem-lower-bound-SSS}(b), that the SRE lengths are
unnecessary when the rightmost occurrence of every SRE is selected.

\begin{lemma}\label{lem-reduction}
For every \(i\ge 1\), \(\mathrm{LRS}[i]=\max\{\,\Delta-1 : \langle p,\Delta\rangle\in \widehat{A}_i \text{ and } p+\Delta-1=i\,\}\).
\end{lemma}

\begin{proof} 
Let \(\Lambda_i=\max\{\,\Delta : \langle p,\Delta \rangle\in \widehat{A}_i \text{ and } p+\Delta-1=i\,\}\).
We prove that $\Lambda_i$ is well-defined and that \(\mathrm{LRS}[i]=\Lambda_i-1\).
If \(i=1\), then by convention, \(\langle 1,1\rangle \in \widehat{A}_1\), and hence \(\Lambda_1=1\). Therefore, \(\mathrm{LRS}[1]=0=\Lambda_1-1\).
Now assume \(i>1\), and let $\ell=LRS[i]$. We first show that there exists a record in $\widehat{A}_i$ ending at $i$; this proves that $\Lambda_i$ is well-defined. If $\ell=0$, then the SRE $(\varepsilon,T[i])$ has a unique represented occurrence, so the record $\langle i,1\rangle$ belongs to $\widehat{A}_i$. If $\ell>0$, let $j<i$ satisfy $T[j-\ell+1..j]=T[i-\ell+1..i]$, and let $u=\overleftarrow{T[i-\ell+1..i]}$. By the maximality of $\ell$, the extension $(u,T[i-\ell])$ occurs only as the prefix of $R_i=\overleftarrow{T[1..i]}$, while $u$ is right-maximal. Hence this extension is an SRE with a unique represented occurrence, so $\langle i-\ell,\ell+1 \rangle\in\widehat{A}_i$. Thus $\Lambda_i$ is well-defined and $\Lambda_i\ge \ell+1$.

It remains to prove that $\mathrm{LRS}[i]\ge \Lambda_i-1$. Let \(\langle p,\Delta\rangle\in \widehat{A}_i\) be a record attaining \(\Lambda_i\), so \(p+\Delta-1=i\) and \(\Delta=\Lambda_i\). The corresponding SRE is \((\overleftarrow{T[p+1..i]}, T[p])\). In particular, \(\overleftarrow{T[p+1..i]}\) is right-maximal in \(R_i\). Hence \(T[p+1..i]\) occurs as a suffix of \(T[1..i]\) and also occurs previously in \(T[1..i-1]\). Therefore, by the definition of \(\mathrm{LRS}[i]\), we have \(\mathrm{LRS}[i]\ge |T[p+1..i]|=\Delta-1=\Lambda_i-1\).
%
Combining the two inequalities gives $\mathrm{LRS}[i]=\Lambda_i-1$
\end{proof}

\begin{example}
Consider $T[1..7] = \text{\$caabaa}$, so that $R_7 = \text{aabaac\$}$. The SREs of $R_7$ are $(\varepsilon, \$)$,
$(a,a)$, $(aa,b)$, and $(aa,c)$, represented by the substrings $\$$, $aa$, $aab$, and $aac$, respectively.
In the coordinate system of $T[1..7]$, the corresponding smallest suffixient set for $R_7$ can be represented by the record set ${\langle1,1\rangle,\langle2,3\rangle,\langle5,3\rangle,\langle6,2\rangle}$.
In particular, both $\langle6,2\rangle$ and $\langle5,3\rangle$ satisfy
$p+\Delta-1=7$, and they yield the candidate values
$\Delta-1=1$ and $\Delta-1=2$, respectively. The suffix $aa$ of
$T[1..7]$ occurs previously, whereas the longer suffix $baa$ does not.
Hence $\mathrm{LRS}[7]=2$.
Taking the maximum over all records ending at position $i=7$ therefore gives $\mathrm{LRS}[7]=2$. 
\end{example}




Lemma~\ref{lem-reduction} and Theorem~\ref{theorem-lower-bound-LRS} imply Theorem~\ref{theorem-lower-bound-SSS}(a), since \(LRS[i]\) can be recovered from a length-annotated SSS using only \(O(\log i)\) additional bits. 


\begin{restatable}{theorem}{lowerBoundSSS} \label{theorem-lower-bound-SSS}
Let \(n\ge 4\). Any deterministic online algorithm that, for every text
\(T[1..n]\) over \(\{\$,\#,0,1\}\) where \(\$\) occurs only at position \(1\),
correctly maintains after every round \(j\in[1..n]\) either of the following
SSS representations for \(\overleftarrow{T[1..j]}\) requires \(\Omega(n)\) bits
of peak working space in the worst case: (a) a length-annotated SSS; or
(b) the rightmost SSS.
\end{restatable}

\begin{proof}
We prove only part~(b) here, since part~(a) has already been established.

Let \(R_i=\overleftarrow{T[1..i]}\) and
\(\ell=\mathrm{LRS}[i]\). Recall that
\({RSSS}(R_i)\) is obtained by selecting, for every SRE of
\(R_i\), the ending position of its rightmost occurrence. We first prove
that \(\mathrm{LRS}[i]=\min{RSSS}(R_i)-1\).
We distinguish two cases according to whether $\mathrm{LRS}[i]=0$.

Suppose first that \(\ell=0\). Then \(R_i[1]=T[i]\) does not occur in
\(T[1..i-1]\), and hence occurs only once in \(R_i\). Therefore,
\((\varepsilon,R_i[1])\) is an SRE whose unique occurrence ends at
position \(1\). Thus, \(1\in{RSSS}(R_i)\), and consequently
\(\min {RSSS}(R_i)=1=\ell+1\).

Now suppose that \(\ell>0\), and let \(u=R_i[1..\ell]\) and
\(a=R_i[\ell+1]=T[i-\ell]\). By the definition of \(\mathrm{LRS}[i]\),
the string \(u\) occurs at least twice in \(R_i\), but the extension \(ua=R_i[1..\ell+1]\) occurs only at the beginning of
\(R_i\).  Therefore, \(u\) is right-maximal and \((u,a)\) is
an SRE whose unique, and hence rightmost, occurrence ends at position
\(\ell+1\). It follows that
\(\ell+1\in {RSSS}(R_i)\).

It remains to show that no position smaller than \(\ell+1\) belongs to
\({RSSS}(R_i)\). Suppose, for contradiction, that some
\(k\le\ell\) belongs to \({RSSS}(R_i)\). Let \(x\) be the
string represented by the corresponding SRE, and let \(d=|x|\). Its
selected rightmost occurrence is \(R_i[k-d+1..k]\), which lies entirely
within \(R_i[1..\ell]\). Since \(R_i[1..\ell]\) has another occurrence
starting at some position \(q>1\), this contradicts
the choice of the occurrence ending at \(k\) as the rightmost occurrence
of \(x\). Therefore, every position in \({RSSS}(R_i)\) is
at least \(\ell+1\), and hence
\(\min {RSSS}(R_i)=\ell+1\). 

Thus, in both cases,
\(\mathrm{LRS}[i]=\min {RSSS}(R_i)-1\).
Given a representation of \({RSSS}(R_i)\), its minimum can
be found by enumerating the represented positions using only
\(O(\log i)\) additional bits. 
An online algorithm maintaining \(RSSS(R_i)\) with \(o(i)\) bits after round \(i\), for all inputs of length \(n\ge 2i\), would yield an online algorithm computing \(LRS[i]\) with \(o(i)+O(\log i)=o(i)\) bits, contradicting Theorem~\ref{theorem-lower-bound-LRS}. Hence, maintaining the rightmost SSS requires \(\Omega(i)\) bits after round \(i\) in the worst case whenever \(n\ge 2i\). Taking \(i=\lfloor n/2\rfloor\) yields the \(\Omega(n)\)-bit peak-working-space bound.
\end{proof}

\section{Conclusion}

In this paper, we studied the online computation of the longest repeating suffix and the online construction of smallest suffixient sets in compressed working space. We presented two worst-case space-time trade-offs for both problems: $O(r \log n + n)$ bits of working space and $O(\log^2 n / \log \log n)$ time per character, and $O(r \log n + n \log \log n)$ bits of working space and $O((\log n / \log \log n)^2)$ time per character, where $r$ is the number of runs in the BWT of the reverse of the text.


We also established a direct connection between the two problems through
reductions from online LRS computation to online maintenance of smallest
suffixient sets. We proved that any deterministic online algorithm for
computing LRS requires $\Omega(n)$ bits of peak working space in the worst
case, even over a constant-size alphabet. Through these reductions, the same
lower bound applies to maintaining either an arbitrary length-annotated
smallest suffixient set or the position-only rightmost smallest suffixient set.

Several questions remain open. One natural direction is to investigate whether the extra $\log\log n$ factor in the second trade-off is necessary: can one achieve $O((\log n/\log\log n)^2)$ worst-case time per round while using only $O(r\log n+n)$ bits of working space?
Another direction is to investigate whether similar run-length BWT-based techniques can support other online string-processing tasks that are currently handled through explicit suffix-tree constructions.



\bibliography{main-arxiv}

\end{document}